\documentclass[lettersize,journal]{IEEEtran}
\usepackage{amsmath,amsfonts,amssymb,amssymb}
\usepackage{algorithmicx}
\usepackage{algpseudocode}
\usepackage{threeparttable}
\usepackage{epsfig}
\usepackage{algorithm}
\usepackage{array}
\usepackage{amsthm}
\usepackage{amsmath}  
\usepackage[caption=false,font=footnotesize]{subfig}
\usepackage{textcomp}
\usepackage{stfloats}
\usepackage{url}
\usepackage{verbatim}
\usepackage{graphicx}
\usepackage{cite}
\usepackage{tablefootnote}
\usepackage{titlesec}
\usepackage{xcolor}

\usepackage{subcaption}

\titlespacing*{\section}{0pt}{\dimexpr10.5pt-3pt}{\dimexpr8pt-5pt}  
\titlespacing{\subsection}{0pt}{2pt}{3pt}

\newtheorem{corollary}{Corollary}

\newtheorem{mylemma}{Lemma}
\graphicspath{{pic/},{pics/}}

\allowdisplaybreaks
    \title{ Hierarchical Codebook Design and Low-Overhead Beam Training for Near-Field Communications With Uniform Circular Arrays}

\author{Gen~Luo,~Hang~Yuan, Xiaozheng~Gao,~\IEEEmembership{Member,~IEEE}, \\Minwei~Shi,~Chong~Han,~\IEEEmembership{Senior~Member,~IEEE}, Kai~Yang,~\IEEEmembership{Member,~IEEE}
	
	\thanks{G.~Luo, H.~Yuan, X.~Gao, and K.~Yang are with the School of Information and Electronics, Beijing Institute of Technology, Beijing 100081,
		China~(email: gen.luo@bit.edu.cn, yuanhang@bit.edu.cn, gaoxiaozheng@bit.edu.cn, yangkai@ieee.org).} 
	\thanks{M.~Shi is with the School of Cyberspace Science and Technology, Beijing Institute of Technology, Beijing 100081, China~(e-mail: shiminwei@bit.edu.cn).}
	\thanks{C.~Han is with the Terahertz Wireless Communications (TWC) Laboratory, Shanghai Jiao Tong University, Shanghai 200240, China~(e-mail: chong.han@sjtu.edu.cn).}

}
\begin{document}

\maketitle

\begin{abstract}


Extremely large-scale multiple-input multiple-output (XL-MIMO) enables near-field location-specific beam focusing for sixth-generation (6G) communications. 
Uniform circular arrays (UCAs), with rotational symmetry and uniform azimuth coverage, have emerged as a key enabling architecture for near-field XL-MIMO systems.
In this paper, we propose a resolution-aware hierarchical codebook for near-field UCA systems, along with an efficient two-stage beam training scheme to significantly reduce the training overhead.
Specifically, we characterize the minimum resolvable distance of UCA systems in the near-field region based on a geometric spherical-wave propagation model, revealing their spatial resolution capability in the joint angle--distance domain. 
Guided by this result, we design a UCA-specific hierarchical codebook, where a power-efficient distance-robust beamforming (DRBF) codebook provides coarse azimuth localization and a full-precision (FP) codebook sampled according to the minimum resolvable distance enables refined angle--distance beam search.
The regularized modal compensation suppresses weak-mode amplification and provides a controllable tradeoff between absolute amplitude gain under unit-norm transmission and distance robustness.
Based on this codebook, we develop a hierarchical decoupled-architecture Bayesian regression (HDA-BAR) scheme for fast and accurate near-field beam training. 
For the considered array configuration, the resulting HDA-BAR training procedure requires $384$ probing slots, corresponding to an approximately \(99.66\%\) overhead reduction relative to the conventional near-field exhaustive-search benchmark.

\end{abstract}


\section{Introduction}

Emerging sixth-generation (6G) applications, including the Internet of Everything (IoE), motivate extremely large-scale multiple-input multiple-output (XL-MIMO) as a key technology for improving spectral efficiency through highly directional transmission~\cite{Dong2025,Huang2025EnergyEfficient,Ray2021,Wang2023}.
In millimeter-wave and terahertz XL-MIMO systems, the enlarged array aperture and shortened wavelength significantly extend the near-field region~\cite{Lu2023,Hou2025,Sun2025,Mingyao2023,Shen2026FullyDynamic}. 
Consequently, conventional far-field plane-wave models and discrete Fourier transform (DFT)-based angular codebooks become inadequate~\cite{Lee2016}.
Instead, spherical-wave channel models are required, under which near-field beamforming can jointly exploit angular and distance information to realize location-specific energy focusing~\cite{Zhang2024,Song2026NearFieldTHzCovert}, thereby providing additional distance-domain resolution and spatial multiplexing gains~\cite{Wei2022,Zhang2022,Zhao2026AIEmpowered}. Thus, near-field-tailored codebook design and efficient beam training are essential for fully exploiting the potential of XL-MIMO systems.

\subsection{Prior Works}
In near-field channels, the array response is jointly determined by angle and distance, leading to inherent angle--distance coupling.
For uniform linear array (ULA)-based near-field beam training, the polar-domain codebook discretizes the channel space in the joint angular--range domain, with each beamforming vector associated with a prescribed angular--distance sample. In this construction, the angular dimension is uniformly quantized, whereas the distance samples are arranged in a non-uniform manner to match the near-field focusing characteristics~\cite{Cui2022ULA}.
However, due to the additional distance dimension, exhaustive beam training based on such codebooks incurs prohibitively high overhead, which scales with the product of the angular and distance sampling sizes~\cite{CodebookWei2022}.
To alleviate this burden, two-phase near-field beam training schemes have been introduced, where the angular direction is first identified and the range is then refined~\cite{TrainingZhang2022}.
Coarse-to-fine beam training schemes have also been investigated, where low-resolution beams are first used for coarse localization and high-resolution beams are subsequently employed for beam refinement~\cite{Qi2022,Lu2024,Yuanwei2024,Shi2024}.
More recent efforts have explored inference-based near-field beam training, where the correlation among near-field codewords is leveraged to identify the optimal beam with only a few beam probes~\cite{Zhuo2025}.



Although these approaches effectively reduce the training overhead of near-field beam search, most of them are developed for ULA-based systems. 
However, ULAs suffer from angle-dependent effective aperture shrinkage in the near-field region, especially for users located at large incidence or departure angles, which leads to non-uniform near-field coverage~\cite{CuiDai2024}.
In contrast, uniform circular arrays (UCAs), owing to their geometrically rotationally symmetric structure, are expected to provide more uniform array responses across different azimuth directions~\cite{Xie2024UCA,Guo2024UCA}.
From an implementation perspective, this rotational symmetry also facilitates uniform sectorization, full-azimuth processing, and codebook reuse under azimuth rotation.
Thus, UCAs have attracted considerable attention in near-field XL-MIMO systems. 
In particular, it was shown in~\cite{Wu2024UCA} that, benefiting from rotational symmetry, UCAs can provide a more uniform and broader near-field region over all angular directions compared to ULAs, allowing more users to benefit from near-field communications. 
For near-field UCA systems, a concentric-ring codebook was further proposed in~\cite{Wu2024UCA}, where the codewords are sampled in both the angular and distance domains to characterize the spatial focusing property of near-field channels. 
However, such joint discretization in the angular and distance domains also increases the codebook size rapidly, resulting in high beam training overhead. 

Phase-mode excitation has been exploited in UCA beamforming design to realize low-complexity beamspace processing~\cite{Zhang2019FIBF}. 
Furthermore, a low-dimensional UCA codebook has been developed to reduce the complexity of near-field beam training~\cite{Ding2023UCA}.
However, although this low-dimensional UCA codebook can reduce the training overhead, its beamforming design reduces the number of effective spatial modes, thereby increasing the beamwidth. This may weaken the near-field spatial focusing capability and spatial multiplexing potential of UCAs, making it difficult to fully exploit the additional gains brought by near-field propagation.

Although several studies have investigated near-field beamforming and codebook design for UCAs, the fundamental spatial resolution capability of UCAs, i.e., the ability to resolve users with different azimuth angles and propagation distances in the near-field region, remains insufficiently characterized.
This capability directly reflects the level of spatial focusing, which is crucial for resolution-aware codebook design, as it provides a quantitative basis for codebook spatial sampling.
Motivated by this, the design of UCA-specific near-field codebooks that fully exploit the joint angle--distance spatial resolution capability of UCAs and the development of corresponding low-overhead beam training methods deserve further investigation.

\subsection{Our Contributions}



Motivated by these challenges, we first investigate the near-field spatial resolution of UCAs over the joint angle--distance domain and then construct a hierarchical codebook tailored to the resulting resolution characteristics.
Based on this codebook, we further design an efficient Bayesian regression (BAR)-based hierarchical beam training scheme. 
The main contributions of this paper are summarized as follows.

\begin{enumerate}
	\item 
	We characterize the spatial resolution of near-field UCA systems in the joint angle--distance domain. 
	Starting from a geometric spherical-wave propagation model, we derive an approximate closed-form expression for the minimum resolvable distance. 
	This result reveals the local angle--distance separability of UCAs and provides a theoretical basis for near-field codebook sampling.
	
	\item 

	We design a hierarchical codebook specifically for near-field UCA systems. 
	In Layer-1, building upon the phase-mode representation of the UCA response, we propose a power-efficient DRBF codebook that suppresses weak-mode amplification through regularized modal compensation and maximizes the worst-case absolute amplitude gain over each coarse angle--distance sector subject to a distance-ripple constraint. The retained phase-mode order sets the nominal angular bandwidth, while exact-manifold verification determines the final sectorization. 
	In Layer-2, a full-precision (FP) codebook is constructed by sampling the angular and radial domains according to the derived minimum resolvable distance, enabling high-resolution beam refinement. 
	This hierarchical design can reduce the search space while preserving high-resolution near-field focusing capability.
	
	\item 
	We develop a hierarchical decoupled-architecture Bayesian regression (HDA-BAR) beam training scheme based on the proposed hierarchical codebook. 
	In Stage-1, the DRBF codebook is used for coarse azimuth localization. 
	In Stage-2, BAR-based fine search is performed over the pruned full-precision codebook, where a proposed diagonal angular--range coupling (DAC) kernel exploits the structured correlation among near-field codewords for efficient optimal beam identification. This scheme enables accurate near-field beam identification with substantially reduced training overhead.
	\item
	 Numerical simulations are conducted to validate the accuracy of our analysis and to evaluate the effectiveness of the proposed codebook design and beam training scheme. 
	Simulation results show that the achievable rate of the proposed HDA-BAR scheme approaches that of exhaustive search as the reference SNR increases. By exploiting hierarchical candidate pruning and correlation-aware BAR search, the scheme reduces the total beam-training overhead by \(99.66\%\) relative to exhaustive search under the considered configuration.
\end{enumerate}


\textbf{Organization:} 
The rest of this paper is structured as follows.
Section II introduces the system model. 
Section III analyzes the near-field spatial resolution of the UCA and derives the minimum resolvable distance. 
Section IV presents the proposed two-layer hierarchical codebook design. 
Section V develops the proposed HDA-BAR beam training scheme. 
Simulation results are provided in Section VI, followed by conclusions in Section VII.

\textbf{Notation:} 
Boldface lowercase and uppercase letters denote vectors and matrices, respectively. 
\((\cdot)^T\), \((\cdot)^H\), and \((\cdot)^*\) represent the transpose, Hermitian transpose, and complex conjugate, respectively. 
\(\|\cdot\|_2\) denotes the Euclidean norm, and \(\mathbb{E}[\cdot]\) denotes expectation. 
$\mathcal{CN}(0,\sigma^2)$ denotes the zero-mean circularly symmetric complex Gaussian distribution with variance $\sigma^2$.
The imaginary unit is denoted by $j$, and $\mathbf I_p$ denotes the $p\times p$ identity matrix.

\section{System Model}


\label{sec:System Model}
\begin{figure}
	\centering
	\includegraphics[width=0.9\linewidth]{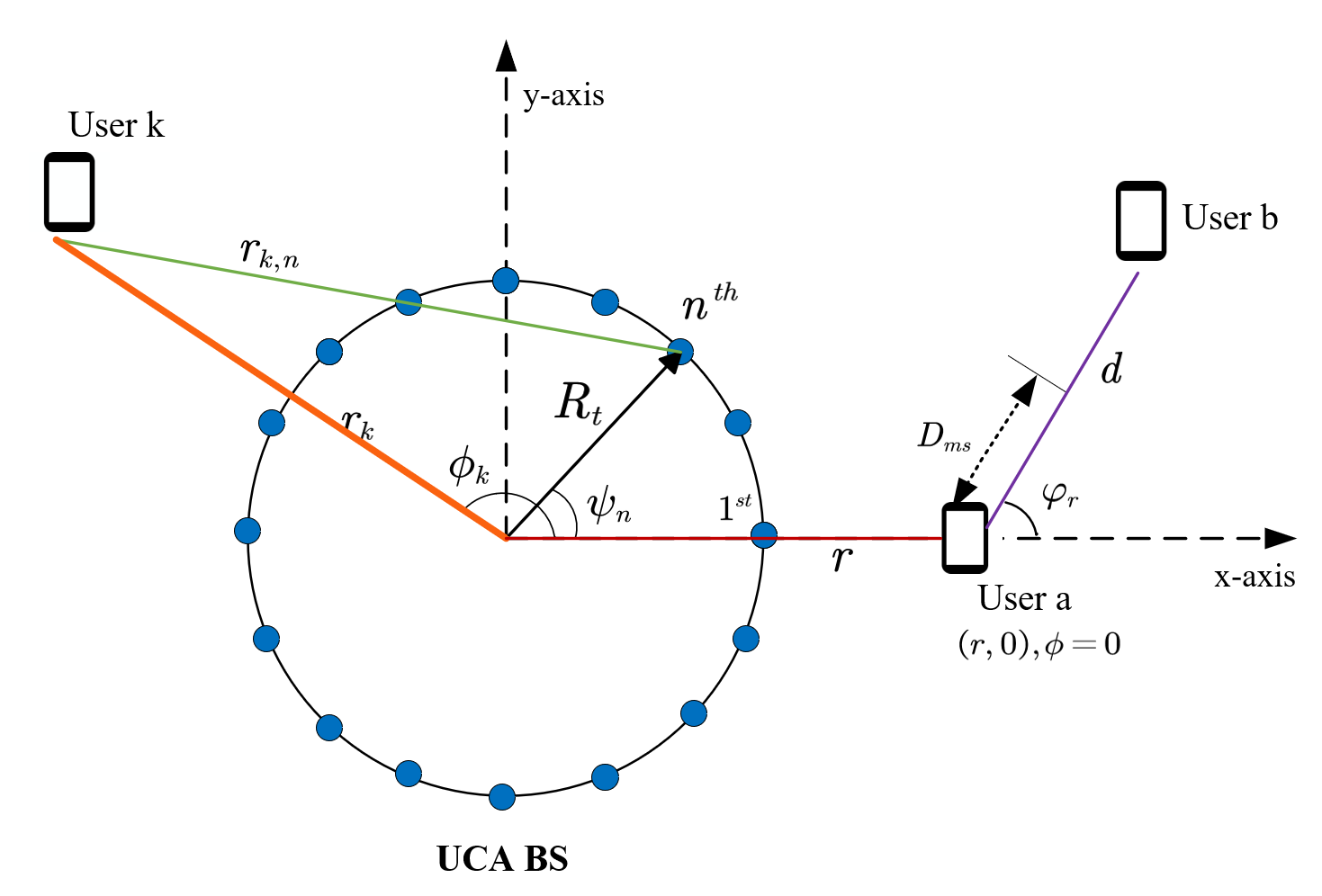}
	\caption{\footnotesize Geometric relationship between the UCA and the near-field user}
	\label{fig:picmode7}
\end{figure}

We consider a narrowband downlink millimeter-wave (mmWave) communication system.
We adopt a narrowband model to focus on spatial resolution and codebook design; wideband extensions accounting for frequency-dependent array responses and beam squint are left for future work.
As illustrated in Fig.~\ref{fig:picmode7}, a base station (BS) equipped with a UCA consisting of \(N\) antennas communicates with \(K\) single-antenna users. 
For ease of exposition, all users are assumed to be located on the same 2D plane as the UCA.

Let \(\mathcal{K}=\{1,2,\ldots,K\}\) denote the user set. 
The location of User~\(k\), \(k\in\mathcal{K}\), is specified by the polar coordinate pair \((r_k,\phi_k)\), where \(\phi_k\) denotes the azimuth angle with respect to the array center and \(r_k\) represents the distance from User~\(k\) to the array center. 
The BS employs a UCA with radius \(R_t\), where the antenna elements are uniformly distributed along the circular aperture. 
Specifically, the \(n\)-th antenna element is located at \((R_t,\psi_n)\), where \(\psi_n=\frac{2\pi(n-1)}{N}\), \(n=1,2,\ldots,N\).
Accordingly, the received signal at User~\(k\) is expressed as
\begin{align}
	y_k
	&=
	\mathbf{h}_k^H\mathbf{w}_k s_k
	+
	\sum_{\substack{v=1, v\neq k}}^{K}
	\mathbf{h}_k^H\mathbf{w}_v s_v
	+
	n_k,
	\label{received_signal_multiuser}
\end{align}

where $\mathbf{h}_k\in\mathbb{C}^{N\times1}$ denotes the downlink channel vector between the BS and User~$k$, $\mathbf{w}_k\in\mathbb{C}^{N\times1}$ denotes the transmit beamforming vector designed for User~$k$, $s_k\in\mathbb{C}$ is the data symbol intended for User~$k$, and $n_k\sim\mathcal{CN}(0,\sigma^2)$ represents the AWGN at User~$k$.
Every training codeword is normalized once to satisfy $\|\mathbf{w}_k\|_2=1$. Following the amplitude-gain convention adopted in phase-mode UCA beamforming, we define
\begin{equation}
	g_{\mathrm{abs}}(\phi,r;\mathbf{w})
	\triangleq
	\left|\mathbf{a}^{\mathrm H}(\phi,r)\mathbf{w}\right|,
	\label{eq:absolute_amplitude_gain}
\end{equation}
where $0\leq g_{\mathrm{abs}}\leq1$. Unless otherwise stated, beamforming gain refers to this amplitude quantity and is reported as $20\log_{10}g_{\mathrm{abs}}$ in decibels. No distance-dependent or pointwise normalization is applied. Since $\mathbf{h}_k=\sqrt{N}\rho_k\mathbf{a}(\phi_k,r_k)$, the corresponding coherent array-power gain in the link budget is $Ng_{\mathrm{abs}}^2$. We assume an active fully digital array, or an equivalent architecture with per-element complex-gain control, so that nonconstant-modulus codewords can be implemented.
The term \(\mathbf{h}_k^H\mathbf{w}_k s_k\) represents the desired signal received by User~\(k\), while the second term in~\eqref{received_signal_multiuser} denotes the multi-user interference caused by the signals intended for the other users.
This interference becomes more pronounced when users are closely located and their channel orthogonality is weak. 
In contrast, when users are sufficiently separated in the spatial domain, the inter-user channel correlation becomes weak, and the interference impact is relatively small.

Since mmWave links are often dominated by the LoS component because of the high path loss and limited diffraction and penetration effects~\cite{Rappaport2019}, we focus on a LoS-dominant channel model to isolate the array-dependent spatial-resolution behavior. The resulting resolution analysis characterizes the dominant LoS component; additional multipath components introduce further angle--distance components and are beyond the scope of this work. Accordingly, the channel between the BS and User~\(k\) is expressed as
\begin{equation}
	\mathbf{h}_k
	=
	\sqrt{N}\rho_k\mathbf{a}(\phi_k,r_k),
\end{equation}
where \(\rho_k\) denotes the complex path gain, which is modeled based on the Friis transmission formula as
\begin{equation}
	\rho_k
	=
	\frac{\lambda}{4\pi r_k}
	\sqrt{G_tG_r}\,
	e^{-j\frac{2\pi}{\lambda}r_k},
\end{equation}
where \(G_t\) and \(G_r\) denote the transmit and receive antenna gains, respectively. 
The term \(\mathbf{a}(\phi_k,r_k)\) represents the near-field array response vector, which is expressed as
\begin{equation}
	\mathbf{a}(\phi_k,r_k)
	=
	\frac{1}{\sqrt{N}}
	\left[
	e^{-j\frac{2\pi}{\lambda}(r_{k,1}-r_k)},
	\ldots,
	e^{-j\frac{2\pi}{\lambda}(r_{k,N}-r_k)}
	\right]^T,
\end{equation}
where \(r_{k,n}\) denotes the propagation distance between User~\(k\) and the \(n\)-th antenna element of the UCA.

Due to the spherical wavefront in the near-field region, the distances from different antenna elements to the same user are generally distinct. 
Based on the geometric relationship between User~\(k\) and the array, \(r_{k,n}\) is expressed as
\begin{align}
	\label{appromix4}
	r_{k,n}
	&=
	\sqrt{
		r_k^2+R_t^2-2r_kR_t\cos(\phi_k-\psi_n)
	}
	\notag\\
	&=
	r_k
	\sqrt{
		1+
		\frac{
			R_t^2-2r_kR_t\cos(\phi_k-\psi_n)
		}{
			r_k^2
		}
	}
	\notag\\
	&\overset{(a)}{\approx}
	r_k
	-
	R_t\cos(\phi_k-\psi_n)
	+
	\frac{R_t^2}{2r_k}
	\sin^2(\phi_k-\psi_n),
\end{align}
where step \((a)\) is obtained by using the Taylor expansion 
\(\sqrt{1+x}=1+\frac{x}{2}-\frac{x^2}{8}+\mathcal{O}(x^3)\) 
 with $x=\frac{R_t^2-2r_kR_t\cos(\phi_k-\psi_n)}{r_k^2}.$
This approximation is accurate when \(R_t\ll r_k\), which is typically satisfied in the radiative near-field region of 6G millimeter-wave XL-MIMO systems, where the Rayleigh distance may extend to tens of meters while the array radius remains much smaller than the user distance.
In particular, the user distances considered in this paper satisfy \(R_t\ll r_k\leq R_{\mathrm{Ray}}\), where \(R_{\mathrm{Ray}}\) denotes the Rayleigh distance used to distinguish the near-field region from the far-field region.
For a UCA with aperture diameter \(D=2R_t\), the Rayleigh distance is given by \(R_{\mathrm{Ray}}=\frac{2D^2}{\lambda}=\frac{8R_t^2}{\lambda}\).
In the near-field region, wave propagation exhibits spherical characteristics, rendering the conventional far-field assumption underlying angle-based spatial multiplexing no longer valid. 
Accordingly, the multi-user near-field channel is jointly characterized by the angular directions and propagation distances of users~\cite{Bacci2023}. 
This enables the antenna array to distinguish users in both angular and distance domains, even for users with similar angular directions~\cite{Ruiz2024}.
To quantify this enhanced separability, we define the spatial resolution as the ability of the array to distinguish users in the joint angle--distance domain. 
To facilitate subsequent codebook design, we first investigate the spatial resolution of near-field UCA systems, which characterizes their spatial multiplexing capability and multi-user access potential.

\section{Near-Field Spatial Resolution Analysis of UCA}

To characterize the spatial resolution of the considered multi-user near-field UCA system, we analyze the pairwise spatial separability between two representative users selected from the user set \(\mathcal{K}\). 
Specifically, two users, denoted by User~\(a\) and User~\(b\), are considered, where \(a,b\in\mathcal{K}\) and \(a\neq b\). 
From the channel perspective, this separability depends on whether the near-field channel vectors of closely located users can become sufficiently orthogonal in the spatial domain. 
Thus, User~\(b\) is assumed to be located in the vicinity of User~\(a\), and the spatial resolution is quantified by the minimum separation required for their channel vectors to achieve a prescribed level of orthogonality.

User~\(a\) is taken as the reference user. Since the UCA is rotationally symmetric, the analysis is invariant to the absolute azimuth of this reference user.
For convenience, we place User~\(a\) on the \(x\)-axis for analysis, as illustrated in Fig.~\ref{fig:picmode7}. 
Its distance from the array center is denoted by \(r_a=r\), and its azimuth angle is set to \(\phi=0\). The distance between User~\(a\) and the \(n\)-th UCA element can be approximated as
\begin{equation}
	\label{r1n_appromix}
	\ r_{a,n} {\approx}r-R_t\cos(\phi-\psi_n)+\frac{R^2_t}{2r}(1-\cos^2(\phi-\psi_n)).
\end{equation}

Based on this setup, we introduce another User~\(b\) in the vicinity of User~\(a\). 
Let \(d\) denote the distance between User~\(a\) and User~\(b\), and let \(\varphi_r\) denote the angle between the displacement direction from User~\(a\) to User~\(b\) and the radial direction of User~\(a\). 
Since User~\(a\) is placed on the \(x\)-axis, this radial direction coincides with the \(x\)-axis.
Let \(r_{b,n}\) denote the distance from User~\(b\) to the \(n\)-th UCA element. It can be approximated as follows.

\begin{align}
	\label{r2n_appromix}
	r_{b,n} 
	& \approx r_{a,n} + \frac{\partial r_{a,n}}{\partial r} \Delta r 
	+ \frac{\partial r_{a,n}}{\partial \phi} \Delta \phi \notag\\
	&\approx r_{a,n}+\left[1-\frac{R_t^2}{2r^2}\sin^2\psi_n
	\right]\Delta r\notag\\&\quad-\left[R_t\sin\psi_n+\frac{R_t^2}{r}\sin\psi_n\cos\psi_n
	\right]\Delta\phi,
\end{align}
where $\Delta r = d\cos\varphi_r$ and $\Delta \phi = \frac{d\sin\varphi_r}{r}$ represent the corresponding radial and angular displacements between the two users.
The first approximation follows from a first-order Taylor expansion around the reference location of User~$a$ $(r,\phi)$, while the second approximation is obtained by substituting $r_{a,n}$ with its expression in~\eqref{r1n_appromix}.

After deriving the propagation distances for User~$a$ and User~$b$, the corresponding channel matrix can be constructed. To quantify the spatial separability of two users in the near-field angle--distance domain, effective degrees of freedom (EDoF) is adopted as a key performance metric. Since the EDoF characterizes the number of effective parallel spatial channels~\cite{Roy2002,Song2026NearField}, the minimum resolvable distance between two users can be identified by examining the variation of the EDoF with the inter-user spacing.
Specifically, the multiuser channel from two users to the UCA is represented by $\mathbf{H} \in \mathbb{C}^{2 \times N}$. Under the assumption that the path gains across different antenna elements are approximately identical~\cite{Liu2023}, the channel variation is mainly governed by the phase differences induced by the propagation distances. Accordingly, $\mathbf{H}$ can be expressed as
\begin{equation}
	\mathbf{H} =
	\begin{bmatrix}
		e^{-j \frac{2\pi}{\lambda} r_{a,1}} & \cdots & e^{-j \frac{2\pi}{\lambda} r_{a,N}} \\
		e^{-j \frac{2\pi}{\lambda} r_{b,1}} & \cdots & e^{-j \frac{2\pi}{\lambda} r_{b,N}}
	\end{bmatrix}.
	\label{H_matrix}
\end{equation}

The gain matrix is denoted by $\mathbf{G}$. This matrix captures the energy coupling between the two users and antenna elements. Based on the singular value decomposition of \(\mathbf{H}\), \(\mathbf{G}\) can be decomposed as
\begin{align}
	\mathbf{G} &= \mathbf{H}\mathbf{H}^H \notag \\
	&= \mathbf{U} \boldsymbol{\Sigma} \mathbf{V}^H \mathbf{V} \boldsymbol{\Sigma} \mathbf{U}^H \notag \\
	&= \mathbf{U} \boldsymbol{\Sigma}^2 \mathbf{U}^H,
\end{align}
where $\boldsymbol{\Sigma} = \mathrm{diag}(\sigma_1, \sigma_2)$ contains the singular values of $\mathbf{H}$, and $\mathbf{U}$ and $\mathbf{V}$ are unitary matrices. 
Denote the two eigenvalues of $\mathbf G$ by $\{\gamma_j\}_{j=1}^{2}$, ordered as $\gamma_1\geq\gamma_2\geq0$. They are equal to the squared singular values of $\mathbf H$, directly linking the channel's energy distribution to its singular value decomposition.

Based on the eigenvalues of $\mathbf{G}$ and following the method in~\cite{Song2024}, EDoF are defined as the minimum number of dominant components required to capture a prescribed fraction of the total eigenvalue energy, i.e.,
\begin{equation}
	\label{Edof_condition}
	\frac{\sum_{j=1}^{\mathfrak{D} (\mathbf{G})} \gamma_j}{\sum_{j=1}^{2} \gamma_j} \geq \tau,
\end{equation}
where $\tau \in (0,1]$ is a prescribed threshold. This criterion effectively quantifies the number of spatially separable channels, thereby reflecting the near-field spatial resolution and multiplexing capability of the UCA system.

\begin{mylemma}\label{Glk_express}
	For the two-user channel matrix $\mathbf{H}\in\mathbb{C}^{2\times N}$, let $G_{l,p}$ denote the $(l,p)$-th entry of the gain matrix $\mathbf{G}\in \mathbb{C}^{2\times 2}$, where $l,p\in\{1,2\}$ Then, $\mathbf{G}_{l,p}$ can be approximated as
	\begin{align}
		\label{Glk_appromix}
		\mathbf{G}_{l,p}
		&\approx
		N e^{j\delta d\alpha}
		J_0(\delta d\beta_1)
		J_0(\delta d\beta_2)
		J_0(\delta d\beta_3),
	\end{align}
	where $\delta = l-p$, $\alpha =-\frac{2\pi}{\lambda}(1-\frac{R^2_t}{4r^2})\cos\varphi_r$, $\beta_1 =\frac{2\pi R_t}{r \lambda}\sin\varphi_r $, $\beta_2 =\frac{\pi R^2_t}{2r^2 \lambda }\cos\varphi_r $, $\beta_3
	=\frac{\pi R_t^2}{r^2\lambda}\sin\varphi_r$, and $J_0(\cdot)$ denotes the 0-order bessel function of the first kind.
\end{mylemma}
\begin{IEEEproof}
	The proof is provided in  Appendix~A.
\end{IEEEproof}
Based on Lemma~\ref{Glk_express}, substituting the closed-form entries of $\mathbf G$ into its characteristic equation yields the following approximations for the ordered eigenvalues defined above:
\begin{equation}
	\begin{aligned}
		\gamma_1 &\approx N\left(1 + J_0(d\beta_1)J_0(d\beta_2)J_0(d\beta_3)\right), \\
		\gamma_2 &\approx N\left(1 - J_0(d\beta_1)J_0(d\beta_2)J_0(d\beta_3)\right).
	\end{aligned}
\end{equation}

Substituting $\gamma_1$ and $\gamma_2$ into \eqref{Edof_condition} yields the condition for $\mathrm{EDoF}=2$, which defines the resolution limit at which the UCA can spatially resolve the two users and is expressed as
\begin{equation}
	\label{boundary_distinguish}
	\ J_0(d\beta_1)J_0(d\beta_2)J_0(d\beta_3)= 2\tau -1.
\end{equation} 

For a prescribed threshold \(\tau\), the minimum resolvable distance \(d_{\mathrm{ms}}(\varphi_r,r)\) at relative azimuth \(\varphi_r\) and reference range \(r\) is the smallest positive separation satisfying~\eqref{boundary_distinguish}, i.e.,
 \begin{align}
	d_{\mathrm{ms}}(\varphi_r,r)
	\!= \!\min\left\{\,\! d>0  \!\;\middle|\;  \!J_0(d\beta_1)J_0(d\beta_2)J_0(d\beta_3)\!=\!2\tau-1 \, \!\right\}.
\end{align}

For analytical tractability, by applying the small-argument approximation of the zeroth-order Bessel function, i.e., 
\(J_0(x)\approx1-\frac{x^2}{4}\) for \(|x|\ll1\), and retaining only the dominant \(O(d^2)\) term, we have
\begin{align}
	J_0(d\beta_1)
	J_0(d\beta_2)
	J_0(d\beta_3)
	&\approx
	1-\frac{d^2}{4}
	\left(
	\beta_1^2+\beta_2^2+\beta_3^2
	\right),
\end{align}
where the higher-order terms of \(d^4\) and above are neglected, since \(d\beta_1\), \(d\beta_2\), and \(d\beta_3\) are sufficiently small in the local resolution region. 
Thus, an approximate closed-form expression for \(d_{\mathrm{ms}}(\varphi_r,r)\) is obtained by
\begin{equation}
	\label{min_d}
	d_{\mathrm{ms}}(\varphi_r,r)
	\approx
	\sqrt{
		\frac{
			8(1-\tau)
		}{
			\beta_1^2+\beta_2^2+\beta_3^2
		}
	}.
\end{equation}

When the distance \(r\) from User~\(a\) to the center of the UCA is fixed, the minimum resolvable distance can be determined for each relative azimuth angle \(\varphi_r\). 
Therefore, \(d_{\mathrm{ms}}(\varphi_r,r)\) characterizes the minimum spatial resolution of the UCA in the joint angle--distance domain.
Equation~\eqref{min_d} shows that \(d_{\mathrm{ms}}\) generally decreases with \(R_t\), increases with \(r\), and varies with \(\varphi_r\) through \(\beta_1\), \(\beta_2\), and \(\beta_3\). Although \(N\) does not appear explicitly in this phase-mode approximation, finite array sampling affects its accuracy. This interpretation applies to the LoS-dominant model.

\section{Proposed Hierarchical Codebook Design}
\label{codebook design: analysis}

In this section, we develop a decoupled spatial-resolution-aware hierarchical codebook for near-field UCA systems. 
The proposed codebook exploits the phase-mode structure of the UCA to alleviate the near-field angle--distance coupling and reduce redundant beam search. 
Specifically, the first layer constructs DRBF beams for efficient coarse angular localization, while the second layer performs resolution-aware refinement in the angle--distance domain according to the approximate closed-form expression of \(d_{\mathrm{ms}}(\varphi_r,r)\) derived in the previous section.

\subsection{Layer-1 Power-Efficient DRBF Codebook}
\label{subsec:drbf}

We propose a power-efficient distance-robust beamforming (DRBF) codebook that combines phase-mode truncation with regularized modal compensation to improve the worst-case absolute amplitude gain while retaining distance robustness under unit-norm transmission. Here, power efficiency refers to improving the coherent array-power gain \(Ng_{\mathrm{abs}}^2\) under a fixed unit-norm transmit constraint, rather than to circuit-level energy efficiency.

Specifically, following the established phase-mode representation for UCA beamforming~\cite{Zhang2019FIBF,Yuan2023} and using the Fresnel approximation in~\eqref{appromix4}, the normalized UCA manifold can be decomposed as
\begin{equation}
	[\mathbf{a}(\phi,r)]_n
	\approx
	\frac{1}{\sqrt N}b_n(\phi)\varrho_n(\phi,r),
	\label{eq:drbf_manifold_decomposition}
\end{equation}
where
\begin{align}
	b_n(\phi)
	&=e^{jk_0R_t\cos(\phi-\psi_n)},\notag\\
	\varrho_n(\phi,r)
	&=\exp\!\left[-j\frac{k_0R_t^2}{2r}
	\sin^2(\phi-\psi_n)\right],
	\label{eq:drbf_range_residual}
\end{align}
where $k_0=2\pi/\lambda$. The factor $b_n(\phi)$ captures the dominant angular phase, whereas $\varrho_n(\phi,r)$ represents the residual distance-dependent wavefront curvature.

Applying the Jacobi--Anger expansion to $b_n(\phi)$ and retaining $|m|\leq M$ gives
\begin{align}
	b_n(\phi)
	&\approx
	\sum_{m=-M}^{M}c_m(\phi)e^{-jm\psi_n},\notag\\
c_m(\phi)
	&=j^mJ_m(k_0R_t)e^{jm\phi}.
	\label{eq:drbf_truncated_phase_modes}
\end{align}
Here, $m$ is the circular-mode index, $M$ is the maximum retained mode order, and $J_m(\cdot)$ is the $m$-th-order Bessel function of the first kind.
Define the phase-mode matrix $\mathbf U_M\in\mathbb C^{N\times(2M+1)}$ by
\begin{equation}
	[\mathbf U_M]_{n,m+M+1}
	=\frac{1}{\sqrt N}e^{-jm\psi_n},
	\qquad m=-M,\ldots,M.
	\label{eq:drbf_phase_mode_matrix}
\end{equation}
For a uniformly sampled UCA, $2M<N$ ensures the alias-free orthogonality
$\mathbf U_M^{\mathrm H}\mathbf U_M=\mathbf I_{2M+1}$. For reference, define a strict inverse-modal baseline for target azimuth $\phi_s$ as
\begin{equation}
	q_{s,m}^{\mathrm{inv}}
	=\frac{1}{c_m^*(\phi_s)},
	\qquad
	\mathbf w_s^{\mathrm{inv}}
	=\frac{\mathbf U_M\mathbf q_s^{\mathrm{inv}}}
	{\|\mathbf U_M\mathbf q_s^{\mathrm{inv}}\|_2}.
	\label{eq:inverse_modal_unit_power_beam}
\end{equation}
where $\mathbf q_s^{\mathrm{inv}}=[q_{s,-M}^{\mathrm{inv}},\ldots,q_{s,M}^{\mathrm{inv}}]^T$. The equality $c_m^*(\phi_s)q_{s,m}^{\mathrm{inv}}=1$ shows that strict inversion equalizes the retained angular modal coefficients. However, since $\varrho_n(\phi,r)$ is not included in the phase-mode synthesis, its distance-dependent curvature is not explicitly compensated. It remains to determine whether truncation to $|m|\leq M$ limits the resulting intermodal phase dispersion and distance-dependent gain loss. The following lemma quantifies both effects.

\begin{mylemma}
\label{lem:drbf_range_robustness}
Let $g_{\mathrm{inv}}(r)\triangleq g_{\mathrm{abs}}(\phi_s,r;\mathbf w_s^{\mathrm{inv}})$. For $r>R_t$, $k_0R_t\gg1$, $M\ll k_0R_t$, and $2M<N$, the leading distance-dependent phase of the $m$-th retained mode, up to a common phase, satisfies
\begin{equation}
	\widetilde\delta_m(r)
	\approx\frac{m^2}{2k_0r},
	\qquad |m|\leq M.
	\label{eq:modal_range_phase}
\end{equation}
If $M^2/(2k_0r)\ll1$, the corresponding leading-order on-axis amplitude gain satisfies
\begin{equation}
	\frac{g_{\mathrm{inv}}(r)}{g_{\mathrm{inv}}(\infty)}
	\approx
	1-
	\frac{M(M+1)[4M(M+1)-3]}
	{360(k_0r)^2}.
	\label{eq:inverse_modal_range_loss}
\end{equation}
Moreover, relative to the conventional far-field phase error $k_0R_t^2/(2r)$, the maximum retained-mode phase $M^2/(2k_0r)$ is only a fraction $(M/(k_0R_t))^2\ll1$. Hence, modal truncation effectively suppresses the dominant distance-dependent phase variation.
\end{mylemma}
\begin{IEEEproof}
The proof is given in Appendix~B.
\end{IEEEproof}

Lemma~\ref{lem:drbf_range_robustness} establishes truncation-induced distance robustness, but strict inversion in~\eqref{eq:inverse_modal_unit_power_beam} amplifies weak Bessel modes. Under unit-norm transmission, modal orthogonality gives
\begin{equation}
	g_{\mathrm{inv}}(\infty)
	=\frac{2M+1}
	{\sqrt{\displaystyle\sum_{m=-M}^{M}|c_m(\phi_s)|^{-2}}}.
	\label{eq:inverse_modal_absolute_gain}
\end{equation}
A small $|c_m(\phi_s)|$ thus enlarges the denominator and reduces the absolute gain, limiting strict inversion in practice. We therefore propose regularized modal compensation to suppress weak-mode amplification while retaining distance robustness:
\begin{equation}
\begin{aligned}
	\widetilde{\mathbf q}_s(\zeta)
	&=\arg\min_{\mathbf q}
	\sum_{m=-M}^{M}|c_m^*(\phi_s)q_m-1|^2
	+\zeta\|\mathbf q\|_2^2,\\[-0.2ex]
	\widetilde q_{s,m}(\zeta)
	&=\frac{c_m(\phi_s)}{|c_m(\phi_s)|^2+\zeta}.
\end{aligned}
	\label{eq:regularized_modal_coefficient}
\end{equation}
The proposed unit-norm DRBF codeword is
\begin{equation}
	\mathbf w_s^{\mathrm{DRBF}}(\zeta)
	=\frac{\mathbf U_M\widetilde{\mathbf q}_s(\zeta)}
	{\|\mathbf U_M\widetilde{\mathbf q}_s(\zeta)\|_2}
	\label{eq:regularized_drbf_codeword}
\end{equation}
The normalized effective modal weights induced by $c_m^*(\phi_s)\widetilde q_{s,m}(\zeta)$ satisfy
\begin{equation}
	\omega_m(\zeta)\propto
	\frac{|c_m(\phi_s)|^2}{|c_m(\phi_s)|^2+\zeta},
	\qquad \sum_{m=-M}^{M}\omega_m(\zeta)=1.
\end{equation}
Let $g_{\mathrm{DRBF}}(r)\triangleq g_{\mathrm{abs}}(\phi_s,r;\mathbf w_s^{\mathrm{DRBF}})$. Using~\eqref{eq:modal_range_phase}, its leading relative on-axis response satisfies
\begin{equation}
\begin{aligned}
	\frac{g_{\mathrm{DRBF}}(r)}{g_{\mathrm{DRBF}}(\infty)}
	&\approx\left|\sum_{m=-M}^{M}\omega_m(\zeta)
	e^{-j\widetilde\delta_m(r)}\right|\\[-0.2ex]
	&\approx1-\frac{\operatorname{Var}_{\omega}(m^2)}{8(k_0r)^2},
\end{aligned}
\label{eq:regularized_drbf_range_loss}
\end{equation}
where $\operatorname{Var}_{\omega}(m^2)$ is the weighted variance under $\omega_m(\zeta)$. Thus, regularization suppresses weak modes while preserving the truncation-induced distance-robustness mechanism.

We use the scale-independent parameterization
\begin{equation}
	\zeta=\zeta_{\mathrm{rel}}
	\max_{|m|\leq M}|c_m(\phi_s)|^2.
	\label{eq:relative_regularization_parameter}
\end{equation}
Let $L$ denote the prescribed number of uniformly rotated Layer-1 codewords, let $\mathcal R=[r_{\min},r_{\max}]$, and define $\mathcal S_L=\{(\Delta\phi,r):|\Delta\phi|\leq\pi/L,\ r\in\mathcal R\}$. The candidate modal-order set $\mathcal M$ satisfies $2M<N$ and is selected according to the distance-sensitivity scale $M^2/(2k_0r_{\min})$ and the required angular coverage. For each $M\in\mathcal M$, $\zeta_{\mathrm{rel}}$ is swept over $\mathcal Z$, which comprises the strict-inversion point $\zeta_{\mathrm{rel}}=0$ and a finite logarithmic grid of positive values terminated when the unit-norm beam response becomes insensitive to further regularization.

Let $g_{M,\zeta_{\mathrm{rel}}}(\Delta\phi,r)$ denote the exact spherical-wave absolute gain. The best $\zeta_{\mathrm{rel}}$ is first found for each $M$, after which the resulting candidates are compared according to
\begin{equation}
\begin{aligned}
	(M^\star,\zeta_{\mathrm{rel}}^\star)
	&=\arg\max_{\substack{M\in\mathcal M\\\zeta_{\mathrm{rel}}\in\mathcal Z}}
	\min_{(\Delta\phi,r)\in\mathcal S_L}g_{M,\zeta_{\mathrm{rel}}}(\Delta\phi,r)\\[-0.2ex]
	\mathrm{s.t.}\quad &
	20\log_{10}\!\max_{r\in\mathcal R}
	g_{M,\zeta_{\mathrm{rel}}}(0,r)\\[-0.2ex]
	&\quad
	{}-20\log_{10}\!\min_{r\in\mathcal R}
	g_{M,\zeta_{\mathrm{rel}}}(0,r)\leq\epsilon_r.
\end{aligned}
	\label{eq:regularized_sector_maxmin}
\end{equation}
By UCA rotational symmetry, this offline search is performed only for one reference sector, and the complete codebook $\mathbf W^{\mathrm{DRBF}}=\{\mathbf w_s^{\mathrm{DRBF}}\}_{s=0}^{L-1}$ is obtained by rotation.

%
\subsection{Layer-2: Full-Precision Codebook}

In this subsection, we develop a high-resolution full-precision codebook by selecting the angular and radial sampling intervals according to the minimum resolvable distance \(d_{\mathrm{ms}}(\varphi_r,r)\) derived in Section~III.

%
%

\begin{corollary}\label{corollary_3dB}
	Consider a densely sampled UCA under the Fresnel condition \(R_t/r\ll1\) and a local separation \(d/r\ll1\), such that \(|d\beta_1|\), \(|d\beta_2|\), and \(|d\beta_3|\) are sufficiently small. When the near-field beamforming vector
	\(\mathbf{a}(\phi+\Delta\phi,r+\Delta r)\) is employed to serve a user located at \((\phi,r)\),
	the corresponding near-field beamforming gain can be approximated as
	\begin{align}
		g^{\mathrm{NF}}(\phi,r,\phi+\Delta\phi,r+\Delta r)
		&\approx \left|J_0(d\beta_1)J_0(d\beta_2)J_0(d\beta_3)\right|,
	\end{align}
	where \(\Delta r=d\cos\varphi_r\), \(\Delta\phi=d\sin\varphi_r/r\), and \(\beta_1\), \(\beta_2\), and \(\beta_3\) are defined in Lemma~\ref{Glk_express}.
\end{corollary}

\begin{IEEEproof}
	The gain equals the magnitude of the normalized manifold correlation \(\mathbf a^H(\phi,r)\mathbf a(\phi+\Delta\phi,r+\Delta r)\). Substituting the local path-difference expansion in~\eqref{r2n_appromix}, removing the common phase, and applying the Jacobi--Anger expansion and UCA phase-mode orthogonality as in Appendix~A yields the stated result. Under the conditions of the corollary, the higher-order cross-mode and finite-array aliasing terms are negligible, leaving the dominant zero-order Bessel contribution.
\end{IEEEproof}

The normalized manifold correlation in Corollary~\ref{corollary_3dB} is also the amplitude gain of a near-field matched beam steered to the displaced location. Moreover, the EDoF boundary in~\eqref{boundary_distinguish} corresponds to the amplitude-correlation level \(2\tau-1\). Therefore, the EDoF threshold can be mapped consistently to an amplitude-gain criterion for determining the sampling intervals of the FP codebook.
Within the local regime of Corollary~\ref{corollary_3dB}, choosing \(\tau=0.75\) gives \(2\tau-1=0.5\), and hence
\(\left|J_0(d\beta_1)J_0(d\beta_2)J_0(d\beta_3)\right|=0.5\). 
Thus, the boundary induced by the EDoF threshold \(\tau=0.75\) is equivalent to the half-amplitude sampling rule adopted in existing codebook designs~\cite{Wu2024UCA}. The resulting \(d_{\mathrm{ms}}(\varphi_r,r)\) is therefore used to determine the FP-codebook sampling intervals. This boundary corresponds to one-quarter of the peak power rather than a half-power threshold.
Next, we discretize the near-field UCA beam space along the angular and radial dimensions~\cite{Wu2024UCA}, with sampling intervals determined by \(d_{\mathrm{ms}}(\varphi_r,r)\) under the prescribed resolution threshold.
By leveraging the closed-form expression of \(d_{\mathrm{ms}}(\varphi_r,r)\) in~\eqref{min_d}, the sampling steps in the angular and radial domains are given by
\begin{equation}
	\label{sampling_step}
	\begin{cases}
		\displaystyle
		\Delta\phi_\tau
		=
		\frac{\sqrt{8(1-\tau)}\lambda}
		{\pi R_t\sqrt{4+\frac{R_t^2}{r^2}}}
		\approx
		\frac{\sqrt{8(1-\tau)}\lambda}{2\pi R_t},
		\\[3mm]
		\displaystyle
		\Delta r_\tau
		=
		\frac{2\sqrt{8(1-\tau)}\lambda r^2}{\pi R_t^2}.
	\end{cases}
\end{equation}

The angular and radial expressions in~\eqref{sampling_step} follow by setting \(\varphi_r=\pi/2\) for pure angular separation and \(\varphi_r=0\) for pure radial separation, respectively. To obtain a uniform angular grid independent of distance, the codebook construction below adopts the distance-independent approximation of \(\Delta\phi_\tau\) in~\eqref{sampling_step}.

To examine joint coverage, substituting \(\Delta r=d\cos\varphi_r\) and \(\Delta\phi=d\sin\varphi_r/r\) into~\eqref{min_d} yields
\begin{equation}
\begin{aligned}
	g^{\mathrm{NF}}(\Delta\phi,\Delta r;r)
	&\approx 1-\frac{\pi^2R_t^2}{4\lambda^2}\\[-0.3ex]
	&\quad\times\left[\left(4+\frac{R_t^2}{r^2}\right)(\Delta\phi)^2
	+\frac{R_t^2}{4r^4}(\Delta r)^2\right],\\[-0.2ex]
	g_{\mathrm{corner}}^{\mathrm{2D}}(r)
	&\approx1-(1-\tau)\left(1+\frac{R_t^2}{8r^2}\right)\\[-0.3ex]
	&\geq 2\tau-1,\quad r>R_t.
\end{aligned}
\label{eq:joint_2d_coverage_bound}
\end{equation}
For the distance-dependent angular interval in~\eqref{sampling_step}, half-cell offsets give \(g_{\mathrm{corner}}^{\mathrm{2D}}\approx\tau\); the displayed conservative bound applies to the adopted uniform grid. Hence, the Cartesian product has an explicit joint local guarantee.

Let \(L_1\) and \(L_2\) denote the resulting numbers of angular and radial grid points, respectively, with index sets \(\Xi_\phi=\{0,\ldots,L_1-1\}\) and \(\Xi_r=\{0,\ldots,L_2-1\}\). The resulting FP codebook is the Cartesian product
\begin{equation}
	\mathbf{W}^{\mathrm{FP}} =
	\left\{
	\mathbf{a}(\phi_{i_\phi},r_{i_r})
	\mid i_\phi\in\Xi_\phi,\ i_r\in\Xi_r
	\right\},
\end{equation}
and its size is \(N_c=L_1L_2\).
To validate coverage without the local approximation, define the exact-spherical worst-case gain over \(\mathcal D_{\mathrm{FP}}=[0,2\pi)\times[r_{\min},r_{\max}]\) as
\begin{equation}
	g_{\min}^{\mathrm{2D}}
	=\min_{(\phi,r)\in\mathcal D_{\mathrm{FP}}}
	\max_{\mathbf w_i^{\mathrm{FP}}\in\mathbf W^{\mathrm{FP}}}
	\left|\mathbf a^{\mathrm H}(\phi,r)\mathbf w_i^{\mathrm{FP}}\right|.
\label{eq:exact_2d_coverage_metric}
\end{equation}
The exact criterion is \(g_{\min}^{\mathrm{2D}}\geq2\tau-1\), i.e., \(g_{\min}^{\mathrm{2D}}\geq0.5\) for \(\tau=0.75\). Unlike the local bound, this metric evaluates angle--distance coupling and every cell interior without pointwise normalization. The construction is summarized in \textbf{Algorithm~\ref{alg:fp_codebook}}.
The construction has complexity $\mathcal{O}(NL_1L_2)$ and can be performed offline or over a local candidate region; its resolution is controlled by $\tau$.

\begin{algorithm}[h]
	\caption{Construction of the FP codebook \(\mathbf{W}^{\mathrm{FP}}\)}\label{alg:fp_codebook}
	\begin{algorithmic}[1]
		\State \textbf{Input:} Minimum distance \(r_{\min}\); maximum distance \(r_{\max}\); circular array radius \(R_t\); wavelength \(\lambda\); resolution threshold \(\tau\).
		
		\State Set $\Delta\phi_\tau=\frac{\sqrt{8(1-\tau)}\lambda}{2\pi R_t}$.
		\State Calculate the angular sampling size $L_1=\left\lceil\frac{2\pi}{\Delta\phi_\tau}\right\rceil$.
		\State Set \(\Xi_\phi=\{0,\ldots,L_1-1\}\) and \(\phi_{i_\phi}=2\pi i_\phi/L_1\) for every \(i_\phi\in\Xi_\phi\).
		\State Initialize \(r_0\leftarrow r_{\min}\) and \(i_r\leftarrow0\).
	    \While{$r_{i_r}+\frac{2\sqrt{8(1-\tau)}\lambda r_{i_r}^2}{\pi R_t^2}\leq r_{\max}$}
		\State $r_{i_r+1}\leftarrow r_{i_r}+\frac{2\sqrt{8(1-\tau)}\lambda r_{i_r}^2}{\pi R_t^2}$.
		\State $i_r\leftarrow i_r+1$.
		\EndWhile
		\State If $r_{i_r}<r_{\max}$, append $r_{\max}$ and update $i_r\leftarrow i_r+1$.
		\State Set \(L_2=i_r+1\) and \(\Xi_r=\{0,\ldots,L_2-1\}\).
		\State Construct the FP codebook as
			\[
			\mathbf{W}^{\mathrm{FP}}
			=
			\left\{
			\mathbf{a}(\phi_{i_\phi},r_{i_r})
			\mid
			i_\phi\in\Xi_\phi,\,
			i_r\in\Xi_r
			\right\}.
			\]
		\State \textbf{Output:} FP codebook \(\mathbf{W}^{\mathrm{FP}}\).
	\end{algorithmic}
\end{algorithm}

\section{Proposed Near-Field HDA-BAR Beam Training}
This section introduces a 2D near-field HDA-BAR beam training scheme that builds upon the hierarchical codebook developed above.
Although the system model considers multiple users, the beam training procedure is described for an arbitrary User~\(u\in\mathcal{K}\), and the same procedure can be independently applied to other users. 
For notational simplicity, the user index is omitted when no ambiguity arises.
As illustrated in Fig.~\ref{fig:hda3}, the proposed scheme consists of two sequential stages.
In the first stage, the azimuth angle of the considered user is rapidly estimated using the coarse DRBF codebook.
In the second stage, the BAR method is employed over the reduced FP codebook to efficiently identify the optimal near-field beam with high accuracy.

\begin{figure}[t]
	\centering
	\includegraphics[width=0.85\linewidth]{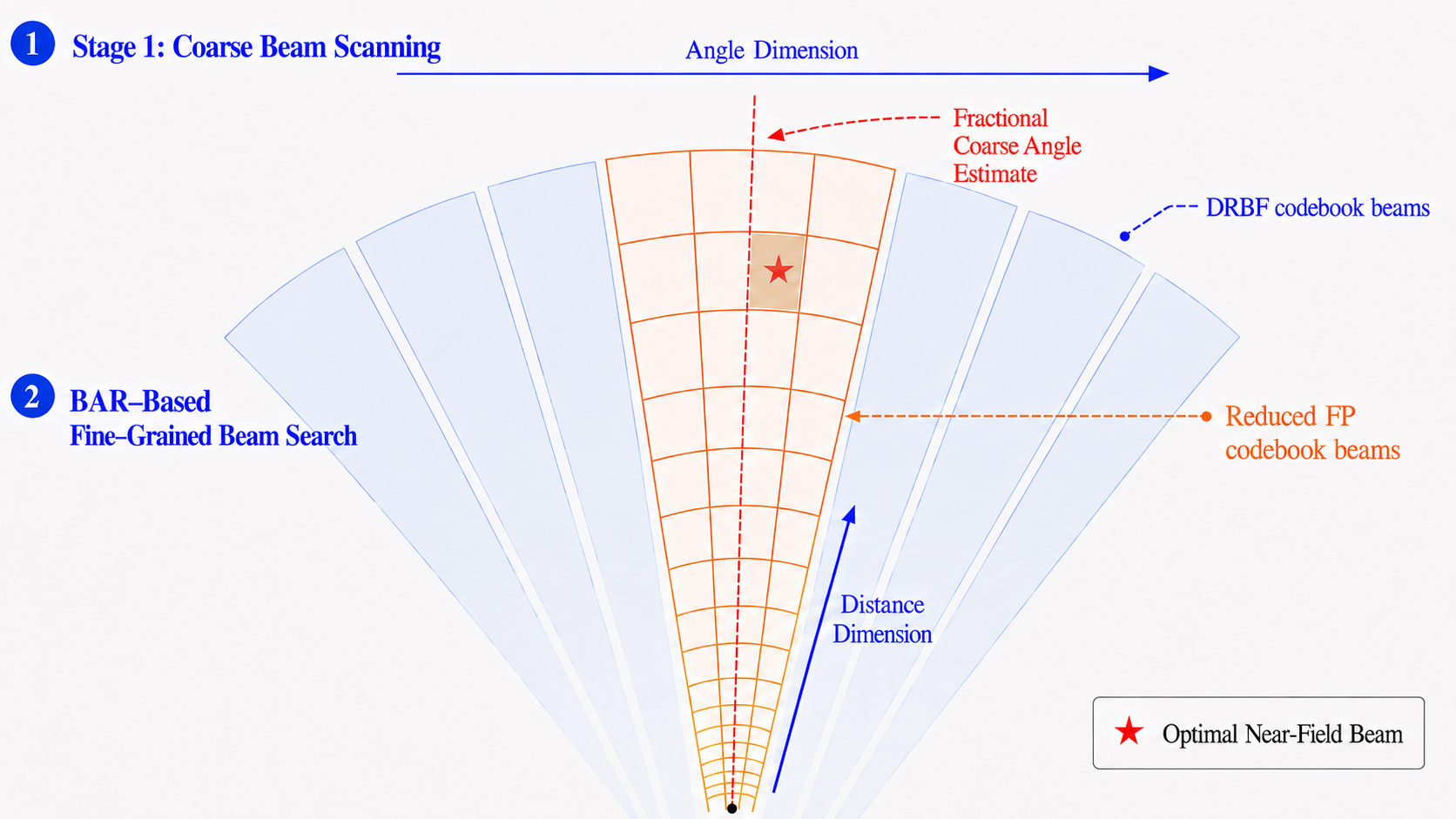}
	\caption{\footnotesize Illustration of the proposed two-stage near-field HDA-BAR beam training scheme.}
	\label{fig:hda3}
\end{figure}

\subsection{Energy-Correlation-Based Fractional Angle Refinement}
\label{subsec:energy_correlation_refinement}

In Stage~1, the $L$ codewords in $\mathbf W^{\mathrm{DRBF}}$, indexed by $\eta=0,\ldots,L-1$, sweep all azimuth sectors with $\phi_\eta=2\pi\eta/L$. Using an orthogonal training pilot of power $P_{\mathrm{tr}}$, the observation and measured energy for the $\eta$-th codeword are
\begin{align}
	y_\eta
	&=\sqrt{P_{\mathrm{tr}}}\,\mathbf h^{\mathrm H}
	\mathbf w_\eta^{\mathrm{DRBF}}+n_\eta,\notag\\
	z_\eta&=|y_\eta|^2,
	\qquad n_\eta\sim\mathcal{CN}(0,\sigma^2).
	\label{eq:stage1_energy_observation}
\end{align}
The maximum-energy codeword provides the sector-level estimate
\begin{equation}
	\eta^\star=\arg\max_{0\leq\eta<L}z_\eta,
	\qquad
	\phi_{\max}=\frac{2\pi\eta^\star}{L}.
	\label{eq:drbf_coarse_direction}
\end{equation}

Owing to UCA rotational symmetry, the local energy responses of different DRBF sectors are circular shifts of a common reference template. Let $Q$ denote the number of neighboring codewords on either side of $\eta^\star$, let $\mathcal Q=\{-Q,\ldots,Q\}$, and let $[\cdot]_L$ denote modulo-$L$ indexing. The corresponding local energy vector is
\begin{equation}
	[\mathbf z_{\mathrm{loc}}]_{q+Q+1}
	=z_{[\eta^\star+q]_L},
	\qquad q\in\mathcal Q.
	\label{eq:local_energy_vector}
\end{equation}
Let $\mathcal E\subseteq[-\pi/L,\pi/L]$ be a fractional-offset grid, and choose $N_r$ offline distance samples $\{r_j\}_{j=1}^{N_r}\subseteq\mathcal R$. Using the exact spherical-wave manifold, define the distance-averaged reference template as
\begin{equation}
	[\mathbf t(\epsilon)]_{q+Q+1}
	=\frac{1}{N_r}\sum_{j=1}^{N_r}
	g_{\mathrm{abs}}^2
	\bigl(\epsilon,r_j;\mathbf w_{[q]_L}^{\mathrm{DRBF}}\bigr),
	\quad q\in\mathcal Q.
	\label{eq:drbf_energy_template}
\end{equation}
Only one reference sector is required. With $\bar z=\mathbf 1^{\mathrm T}\mathbf z_{\mathrm{loc}}/(2Q+1)$ and $\bar t(\epsilon)=\mathbf 1^{\mathrm T}\mathbf t(\epsilon)/(2Q+1)$, define $\check{\mathbf z}_{\mathrm{loc}}=\mathbf z_{\mathrm{loc}}-\bar z\mathbf 1$ and $\check{\mathbf t}(\epsilon)=\mathbf t(\epsilon)-\bar t(\epsilon)\mathbf 1$. Their normalized correlation is
\begin{equation}
	\mathcal C(\epsilon)
	=\frac{\check{\mathbf z}_{\mathrm{loc}}^{\mathrm T}
	\check{\mathbf t}(\epsilon)}
	{\|\check{\mathbf z}_{\mathrm{loc}}\|_2
	\|\check{\mathbf t}(\epsilon)\|_2}.
	\label{eq:energy_template_correlation}
\end{equation}
The fractional offset and the refined azimuth estimate are obtained as
\begin{equation}
	\hat\epsilon_\phi
	=\arg\max_{\epsilon\in\mathcal E}\mathcal C(\epsilon),
	\qquad
	\hat\phi
	=\operatorname{wrap}_{[0,2\pi)}
	\bigl(\phi_{\max}+\hat\epsilon_\phi\bigr).
	\label{eq:energy_correlation_angle_estimate}
\end{equation}
Centering and normalization remove the common energy offset and overall scale. The template set $\mathcal T_{\mathrm{EC}}=\{\mathbf t(\epsilon):\epsilon\in\mathcal E\}$ is generated offline; hence, the refinement requires no additional pilot and incurs $\mathcal O((2Q+1)|\mathcal E|)$ online operations.

\subsection{BAR-Based Fine-Grained Beam Search}
After obtaining the refined azimuth estimate \(\hat{\phi}\), the second-layer FP codebook is pruned in the angular domain. Since exhaustive search over the remaining candidates may still incur considerable training overhead, the proposed BAR search exploits their correlation for efficient beam identification.

Since adjacent Layer-1 DRBF codewords are separated by $2\pi/L$, the pruning half-width is set to $\Delta_\phi=2\pi/L$, thereby retaining one Layer-1 angular interval on either side of \(\hat{\phi}\) without introducing an additional tuning parameter. The reduced codebook is constructed as
\begin{equation}
	\mathbf{W}^{\mathrm{Red}}
	=
	\left[
	\mathbf{w}^{\mathrm{FP}}_i
	\,\middle|\,
	\substack{
	\mathbf{w}^{\mathrm{FP}}_i\in\mathbf{W}^{\mathrm{FP}},\\[-0.2ex]
	\left|\operatorname{wrap}_{[-\pi,\pi)}
	(\phi_i-\hat{\phi})\right|\leq\Delta_\phi
	}
	\right]
	\in
	\mathbb{C}^{N\times N_{\mathrm{red}}},
	\label{W_red}
\end{equation}
where \(\phi_i\) and \(r_i\) denote the angular and radial grid points of \(\mathbf{w}^{\mathrm{FP}}_i\), respectively, and \(N_{\mathrm{red}}\) is the number of retained candidate codewords. The circular angular difference in~\eqref{W_red} ensures continuous pruning across the $0$--$2\pi$ boundary.
The objective is to select the optimal beam from \(\mathbf{W}^{\mathrm{Red}}\) by maximizing the received signal power, i.e.,
\begin{equation}
	\mathbf{w}^{\ast}
	=
	\underset{\mathbf{w}^{\mathrm{Red}}_i\in\mathbf{W}^{\mathrm{Red}}}{\arg\max}
	\;
	\left|
	\mathbf{h}^{\mathrm H}
	\mathbf{w}^{\mathrm{Red}}_i
	\right|^2,
\end{equation}
where \(\mathbf{w}^{\mathrm{Red}}_i\) denotes the \(i\)-th candidate codeword in \(\mathbf{W}^{\mathrm{Red}}\).
Based on the reduced candidate codebook \(\mathbf{W}^{\mathrm{Red}}\), we propose a BAR-based beam training scheme to identify the optimal beam efficiently.
By exploiting the strong correlation among the candidate codewords in \(\mathbf{W}^{\mathrm{Red}}\), the posterior distribution is iteratively updated with a limited number of observations.
At each iteration, the next codeword is adaptively selected according to both the predictive mean and the associated uncertainty.

For the candidate codewords \(\mathbf{w}^{\mathrm{Red}}_i\), 
\(i=1,2,\ldots,N_{\mathrm{red}}\), the noise-free received beam-energy sequence can be written as
\begin{equation}
	p_i^{\mathrm{Red}}
	=
	\left|
	\mathbf{h}^{\mathrm H}
	\mathbf{w}^{\mathrm{Red}}_i
	\right|^2,
	\quad i=1,2,\ldots,N_{\mathrm{red}}.
\end{equation}
Accordingly, the beam-energy vector over the reduced codebook is given by
\begin{equation}
	\mathbf{p}^{\mathrm{Red}}
	=
	\left[
	p_1^{\mathrm{Red}},
	p_2^{\mathrm{Red}},
	\ldots,
		p_{N_{\mathrm{red}}}^{\mathrm{Red}}
	\right]^T,
\end{equation}
where \(\mathbf{p}^{\mathrm{Red}}\in\mathbb{R}^{N_{\mathrm{red}}\times1}\) and 
\(\mathbf{W}^{\mathrm{Red}}\in\mathbb{C}^{N\times N_{\mathrm{red}}}\).
The vector \(\mathbf{p}^{\mathrm{Red}}\) is modeled as a Gaussian process~\cite{Zijian2024,Rasmussen2006GP}, i.e.,
\begin{equation}
	\mathbf{p}^{\mathrm{Red}}
	\sim
	\mathcal{GP}
	\left(
	\mathbf{0},
	\mathbf{K}
	\right),
\end{equation}
where \(\mathbf{K}\in\mathbb{R}^{N_{\mathrm{red}}\times N_{\mathrm{red}}}\) is the covariance matrix. 
In Gaussian process modeling, \(\mathbf{K}\) is constructed using a kernel function that characterizes the correlation between the received beam energies of different candidate codewords.

To characterize the structured beam-energy distribution in near-field UCA systems, we propose a DAC kernel.
Unlike conventional kernels, the DAC kernel captures the fact that the main-lobe energy does not follow a simple elliptical pattern, but instead exhibits an X-shaped main-lobe energy distribution over the angle--distance sampled codebook~\cite{ly2024,Cui2023NearFieldWideband}, as illustrated in Fig.~\ref{fig:XKernel_compare}.
\begin{figure}[t]
	\centering
	\begin{minipage}{0.49\linewidth}
		\centering
		\includegraphics[width=\linewidth]{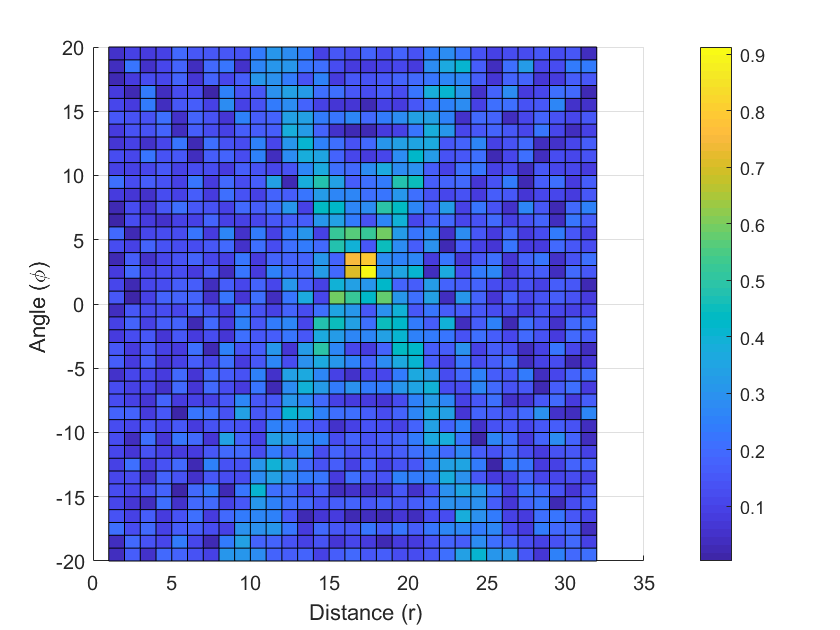}\\
		{\footnotesize (a) Actual energy distribution}
	\end{minipage}
	\hfill
	\begin{minipage}{0.49\linewidth}
		\centering
		\includegraphics[width=\linewidth]{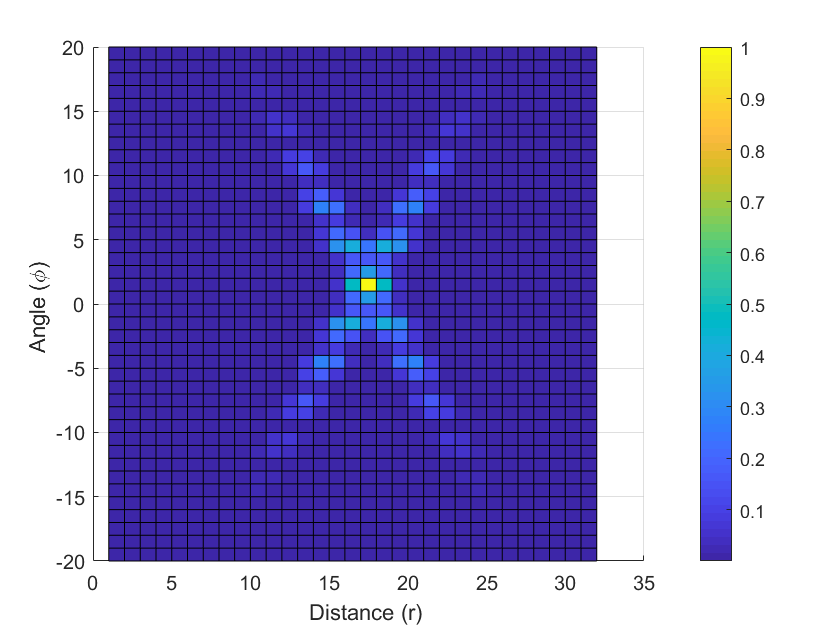}\\
		{\footnotesize (b) Proposed X-shaped DAC kernel}
	\end{minipage}
	\caption{\footnotesize Near-field UCA codebook energy distribution with X-shaped main lobes and the proposed DAC kernel.}
	\label{fig:XKernel_compare}
\end{figure}
The DAC kernel is defined as
\begin{equation}
	K_{\mathrm{DAC}}
	\!\left(
	\boldsymbol{\chi}^{(i)},
	\boldsymbol{\chi}^{(i')}
	\right)
	=
	K_X
	\!\left(
	\boldsymbol{\chi}^{(i)},
	\boldsymbol{\chi}^{(i')}
	\right)
	K_S
	\!\left(
	\boldsymbol{\chi}^{(i)},
	\boldsymbol{\chi}^{(i')}
	\right),
	\label{eq:dac_kernel}
\end{equation}
where
\(\boldsymbol{\chi}^{(i)}=(\chi_{\phi}^{(i)},\chi_r^{(i)})\)
denotes the angular--distance grid index of the \(i\)-th candidate codeword, and \(i'\) denotes another candidate codeword index in the reduced codebook. 
The functions \(K_X\) and \(K_S\) are covariance kernel functions.
Specifically, \(K_X\) captures the X-shaped angular--range correlation induced by near-field propagation, while \(K_S\) enforces local smoothness across the angular--range domain.
They are expressed as
\begin{equation}
	\begin{cases}
		K_X
		\!\left(
		\boldsymbol{\chi}^{(i)},
		\boldsymbol{\chi}^{(i')}
		\right)
		=
		\begin{aligned}[t]
			&\frac{1}{2}
			\exp\!\left(
			-\dfrac{(\Delta\chi_r-\kappa\Delta\chi_{\phi})^2}{2\sigma_x^2}
			\right)
			\\
			&\quad
			+
			\frac{1}{2}
			\exp\!\left(
			-\dfrac{(\Delta\chi_r+\kappa\Delta\chi_{\phi})^2}{2\sigma_x^2}
			\right),
		\end{aligned}
		\\[2mm]
		K_S
		\!\left(
		\boldsymbol{\chi}^{(i)},
		\boldsymbol{\chi}^{(i')}
		\right)
		=
		\exp\!\left(
		-\dfrac{(\Delta\chi_{\phi})^2}{2\sigma_{\phi}^2}
		-
		\dfrac{(\Delta\chi_r)^2}{2\sigma_r^2}
		\right),
	\end{cases}
	\label{eq:kernel_definition}
\end{equation}
where
$
	\Delta\chi_{\phi}
	=
	\chi_{\phi}^{(i)}
	-
	\chi_{\phi}^{(i')}$ and $
	\Delta\chi_r
	=
	\chi_r^{(i)}
	-
	\chi_r^{(i')}.
$
The parameters \(\kappa\), \(\sigma^2_x\), \(\sigma^2_{\phi}\), and \(\sigma^2_r\) are tunable hyperparameters of the DAC kernel, where \(\kappa\) determines the opening angle of the X-shaped correlation pattern and \(\sigma^2_x\), \(\sigma^2_{\phi}\), and \(\sigma^2_r\) control the corresponding correlation widths. Since kernel selection critically affects BAR performance~\cite{Bishop2006}, conventional kernels may be insufficient for near-field UCA MIMO scenarios.
The proposed DAC kernel better matches the near-field beam-energy distribution and thus improves beam selection performance.
Consequently, the covariance matrix is constructed as
\begin{equation}
	\mathbf{K}(i,i')
	=
	K_{\mathrm{DAC}}
	\!\left(
	\boldsymbol{\chi}^{(i)},
	\boldsymbol{\chi}^{(i')}
	\right).
	\label{eq:covariance_kernel}
\end{equation}


Let \(\Omega\) denote the index set of the probed codewords. 
The corresponding real-valued beam-energy observation vector is denoted by 
\(\mathbf{e}_{\Omega}\in\mathbb{R}^{|\Omega|}\), and can be expressed as
\begin{equation}
	\mathbf{e}_{\Omega}
	=
	\mathbf{p}^{\mathrm{Red}}(\Omega)
	+
	\mathbf{n}_{\Omega},
	\quad
	\mathbf{n}_{\Omega}
	\sim
	\mathcal{N}
	\left(
	\mathbf{0}_{|\Omega|},
	\sigma_e^2\mathbf{I}_{|\Omega|}
	\right),
\end{equation}
where \(\mathbf{n}_{\Omega}\) denotes the equivalent real-valued observation noise in the beam-energy domain.
Although the exact perturbation after energy measurement is generally non-Gaussian, it is approximated as Gaussian for tractable GP modeling, which is reasonable in large-scale antenna systems where the SNR after beamforming is relatively high.
Then, the joint distribution of \(\mathbf{p}^{\mathrm{Red}}\) and \(\mathbf{e}_{\Omega}\) is given by
\begin{equation}
	\begin{bmatrix}
		\mathbf{p}^{\mathrm{Red}} \\
		\mathbf{e}_{\Omega}
	\end{bmatrix}
	\sim
	\mathcal{N}
	\left(
	\mathbf{0},
	\begin{bmatrix}
		\mathbf{K}
		&
		\mathbf{K}(:,\Omega)
		\\
		\mathbf{K}(\Omega,:)
		&
		\mathbf{K}(\Omega,\Omega)
		+
		\sigma_e^2\mathbf{I}_{|\Omega|}
	\end{bmatrix}
	\right).
\end{equation}

For the observed codeword index set \(\Omega\), the posterior mean \(\boldsymbol{\mu}_{\Omega}\) and the marginal posterior variance of the \(i\)-th candidate are computed as~\cite{Zhuo2025,Rasmussen2006GP}
\begin{equation}
	\boldsymbol{\mu}_{\Omega}
	=
	\mathbf{K}(:,\Omega)
	\left(
	\mathbf{K}(\Omega,\Omega)
	+
	\sigma_e^2\mathbf{I}_{|\Omega|}
	\right)^{-1}
	\mathbf{e}_{\Omega},
	\label{mean_bar}
\end{equation}
\begin{equation}
	s_{\Omega,i}^{2}
	=
	\mathbf{K}(i,i)
	-
	\mathbf{K}(i,\Omega)
	\left(
	\mathbf{K}(\Omega,\Omega)
	+
	\sigma_e^2\mathbf{I}_{|\Omega|}
	\right)^{-1}
	\mathbf{K}(\Omega,i).
	\label{cov_bar}
\end{equation}

An acquisition function \(V_t(i)\) is then defined to select the next codeword index based on the posterior mean \(\boldsymbol{\mu}_{\Omega_t}\) and standard deviation \(s_{\Omega_t,i}\).
By jointly considering exploitation through the posterior mean and exploration through the posterior uncertainty, the acquisition function is expressed as
\begin{equation}
	V_t(i)
	=
	\mu_{\Omega_t,i}
	+
	\xi s_{\Omega_t,i},
	\label{Vx}
\end{equation}
where $\mu_{\Omega_t,i}$ and $s_{\Omega_t,i}$ denote the posterior mean and standard deviation of the $i$-th candidate codeword at iteration $t$, respectively, and $\xi$ controls the exploitation--exploration tradeoff.
Let \(\mathcal{I}_{\mathrm{red}}=\{1,2,\ldots,N_{\mathrm{red}}\}\) denote the index set of all candidate codewords in the reduced codebook \(\mathbf{W}^{\mathrm{Red}}\).
The next codeword index to probe is selected by maximizing the acquisition function, i.e.,
\begin{equation}
	i_{t+1}
	=
	\arg\max_{i\in\mathcal{I}_{\mathrm{red}}\setminus\Omega_t}
	V_t(i),
	\label{Strategy_i}
\end{equation}
where \(\setminus\) denotes the set difference.
The observation index set is then updated iteratively as
\begin{equation}
	\Omega_{t+1}
	=
	\Omega_t
	\cup
	\{i_{t+1}\}.
\end{equation}

Because Stage~1 provides an angular prior but no reliable range estimate, Stage~2 uses the codeword closest to \(\hat\phi\) on the outermost radial layer as a deterministic initial probe; the subsequent BAR updates explore both angle and range.
The detailed procedure of the proposed HDA-BAR beam training scheme is summarized in \textbf{Algorithm~\ref{alg3}}.

\vspace{0.2em}
\subsection{Complexity and Overhead Analysis}
\vspace{0.15em}

This subsection evaluates the online computational complexity and training overhead of the proposed HDA-BAR beam training scheme.

In {\bf Algorithm~\ref{alg3}}, \(T_{\max}\) denotes the number of Stage~2 pilots, including the initial probe.
Using incremental Cholesky updates, the online computational complexity is
\begin{equation}
\begin{aligned}
\mathcal{O}\bigl(&L+(2Q+1)|\mathcal E|+T_{\max}^{3}\\
&+N_{\mathrm{red}}T_{\max}^{2}
+N_{\mathrm{red}}T_{\max}\bigr).
\end{aligned}
\end{equation}
Here, \(\mathcal{O}(L+(2Q+1)|\mathcal E|)\) denotes the Stage~1 complexity, including DRBF beam sweeping and the energy-template correlation over \(\mathcal E\).
In Stage~2, a rank-one Cholesky update and the evaluation of only the required marginal posterior variances cost \(\mathcal{O}(t^2+N_{\mathrm{red}}t)\) at iteration \(t\), yielding \(\mathcal{O}(T_{\max}^3+N_{\mathrm{red}}T_{\max}^2)\) over the search, while acquisition maximization contributes \(\mathcal{O}(N_{\mathrm{red}}T_{\max})\).
The angularly indexed codebook makes candidate pruning linear in \(N_{\mathrm{red}}\) and hence dominated by the posterior updates. The codebooks, templates, and DAC-kernel entries, whose one-time kernel construction costs \(\mathcal{O}(N_{\mathrm{red}}^2)\), are precomputed offline and excluded from the online complexity.

In addition, beam training overhead is measured by the number of required time slots. The proposed scheme requires $L+T_{\max}$ slots, substantially fewer than those required by conventional search-based schemes, as shown in Table~\ref{tab:overhead_comparison}.


\begin{algorithm}[h]
	\caption{Two-Stage HDA-BAR Beam Training}
	\label{alg3}
	\begin{algorithmic}[1]
		\State \textbf{Input:} Coarse DRBF codebook \(\mathbf{W}^{\mathrm{DRBF}}\); energy-template set \(\mathcal T_{\mathrm{EC}}\); neighbor radius \(Q\); offset grid \(\mathcal E\); near-field FP codebook \(\mathbf{W}^{\mathrm{FP}}\); DAC-kernel covariance matrix \(\mathbf K\); maximum number of pilots \(T_{\max}\).
		
		\State \textbf{Stage 1: Energy-Correlation-Based Coarse Azimuth Refinement}
		\For{\(\eta=0\) to \(|\mathbf{W}^{\mathrm{DRBF}}|-1\)}
			\State Transmit a pilot using the DRBF codeword \(\mathbf{w}^{\mathrm{DRBF}}_{\eta}\).
		\State Compute the received signal energy $z_{\eta}$ for the $\eta$-th codeword.
		\EndFor
		
		\State Find the maximum-energy index
		\[
		{\eta^\star
		=
		\arg\max_{0\leq\eta\leq|\mathbf{W}^{\mathrm{DRBF}}|-1}
		z_{\eta}.}
		\]
		\State Obtain the corresponding azimuth angle \(\phi_{\max}=\phi_{\eta^\star}\).
		\State Form \(\mathbf z_{\mathrm{loc}}\) using~\eqref{eq:local_energy_vector}, evaluate \(\mathcal C(\epsilon)\) using~\eqref{eq:energy_template_correlation}, and obtain \(\hat\phi\) from~\eqref{eq:energy_correlation_angle_estimate}.
		
		\State \textbf{Stage 2: BAR-Based Fine-Grained Beam Search}
		\State Prune \(\mathbf{W}^{\mathrm{FP}}\) according to~\eqref{W_red} using \(\hat{\phi}\), and obtain the reduced codebook \(\mathbf{W}^{\mathrm{Red}}\).
		\State Select the codeword on the outermost radial layer whose azimuth is closest to \(\hat\phi\) as the initial codeword \(\mathbf w^{\mathrm{Red}}_{i_1}\).
		\State Probe \(\mathbf w^{\mathrm{Red}}_{i_1}\), obtain the beam-energy observation \(e_{i_1}\), and initialize \(\Omega_1=\{i_1\}\) and \(\mathbf e_{\Omega_1}=[e_{i_1}]\).

		\For{\(t=1\) to \(T_{\max}-1\)}
		\State Update the posterior means and marginal standard deviations via an incremental Cholesky factorization of~\eqref{mean_bar}--\eqref{cov_bar}.
		\State Select the next codeword index \(i_{t+1}\) according to~\eqref{Vx} and~\eqref{Strategy_i}.

		\State Probe \(\mathbf w^{\mathrm{Red}}_{i_{t+1}}\) and obtain the beam-energy observation \(e_{i_{t+1}}\).
		\State Update \(\Omega_{t+1}\leftarrow\Omega_t\cup\{i_{t+1}\}\) and append \(e_{i_{t+1}}\) to \(\mathbf{e}_{\Omega_{t+1}}\).
		\EndFor
		\State Update the final posterior mean \(\boldsymbol{\mu}_{\Omega_{T_{\max}}}\) with the last observation using the same factorization.

		\State Select the optimal codeword index as
		\[
		i^\star
		=
		\arg\max_{i\in\mathcal{I}_{\mathrm{red}}}
		\mu_{\Omega_{T_{\max}},i}.
		\]
		\State \textbf{Output:} Optimal beamforming codeword \(\mathbf{w}^{\ast}=\mathbf{w}^{\mathrm{Red}}_{i^\star}\).
	\end{algorithmic}
\end{algorithm}

\section{Simulation Results}

This section validates the derived closed-form boundaries of the minimum resolvable region for the UCA configuration through simulations and compares them with those of the ULA.
Numerical results are provided to assess the effectiveness of the hierarchical beam codebook and the HDA-BAR beam training scheme.

We set the number of antennas at the BS as $N=1024$, the carrier frequency as $f=40$ GHz, and the array radius as $R_t=0.96$ m. 
For Layer-1 codebook generation, we prescribe the power-of-two size $L=256$ according to the hierarchical structure and training-overhead budget and set $\epsilon_r=1$~dB. We implement~\eqref{eq:regularized_sector_maxmin} as a two-dimensional exact-manifold grid search: for each candidate $M$, $\zeta_{\mathrm{rel}}$ is optimized first, after which the feasible operating points are compared by their worst absolute amplitude gains over $\mathcal S_L$. The search yields $M^\star=96$ and $\zeta_{\mathrm{rel}}^\star=39.81$ within the tested grid. This operating point retains $193$ modes, satisfies $2M^\star<N$, and achieves a worst absolute amplitude gain of $-13.71$~dB over the assigned angle--distance sector. The corresponding modal-to-aperture phase-sensitivity ratio $(M^\star/(k_0R_t))^2$ is $1.425\times10^{-2}$. Since the small-modal-phase approximation is not tight at the shortest distance, it is used only to explain the sensitivity scaling, whereas the design and coverage evaluation use the exact spherical-wave manifold. We set $\tau=0.75$. For the Stage-1 energy-correlation refinement, we use $Q=4$, a 201-point fractional-offset grid over $[-\pi/L,\pi/L]$, and 41 offline distance samples over $[5,400]$~m.
For Layer~2, Algorithm~\ref{alg:fp_codebook} gives $L_1=3574$ and $L_2=32$ after including the service boundary $r_{\max}=400$~m, yielding $N_c=L_1L_2=114368$ unit-norm FP codewords. With $L=256$, the pruning half-width is $\Delta_\phi=1.40625^\circ$, corresponding to a full angular span of $2.8125^\circ$ and hence $N_{\mathrm{red}}\approx896$.
For simplicity, a single-user scenario is considered, i.e., $K=1$.
In each simulation trial, the user location is randomly generated in the near-field region, where $\phi \in [0,2\pi]$ and $r \in [5,400]$ m. 
For the HDA-BAR beam training scheme, the maximum number of BAR iterations is set to $T_{\max}=128$, the DAC-kernel hyperparameters are set to $\kappa=0.46$, $\sigma_x^2=0.69$, $\sigma_r^2=9$, and $\sigma_\phi^2=2.25$, the equivalent energy-observation noise variance is $\sigma_e^2=0.02$, and the acquisition parameter is set to $\xi=0.92$.
The number of independent Monte Carlo trials is set to $5000$.

\captionsetup[table]{
	labelfont={sc,small},      
	textfont={sc,footnotesize}, 
	labelsep=newline           
}


\begin{table}[!t]
	  \centering%
	\caption{Comparison Of Beam Training Overhead }
	\begin{tabular}{|c|c|c|}
		\hline
		Schemes & Overhead & Value \\
		\hline
		Far-field exhaustive searching scheme & $L_1$ & 3574 \\
		\hline
		Near-field exhaustive searching scheme & $L_1L_2$ & 114368 \\
		\hline
		Near-field TPBT scheme & $L_1$ + $L_2$ & 3606 \\
			\hline
		Proposed HDA-BAR beam training scheme & $L+T_{\max}$ & 384 \\
		\hline
	\end{tabular}
		\label{tab:overhead_comparison}
\end{table}


\begin{figure}[t]
	\centering
	\subfloat[Minimum resolvable distance versus BS-user distance with the antenna aperture fixed at \(1.92\,\mathrm{m}\). ]{
		\includegraphics[width=0.8\linewidth]{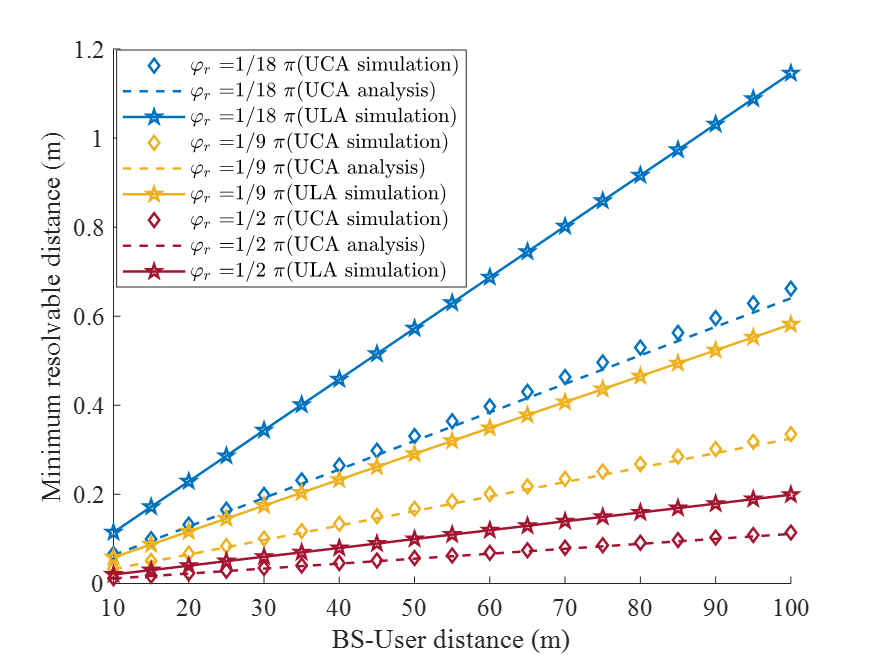}
		
	}\\
	\subfloat[Minimum resolvable distance versus antenna aperture with the BS-user distance fixed at \(15\,\mathrm{m}\).]{
		\includegraphics[width=0.8\linewidth]{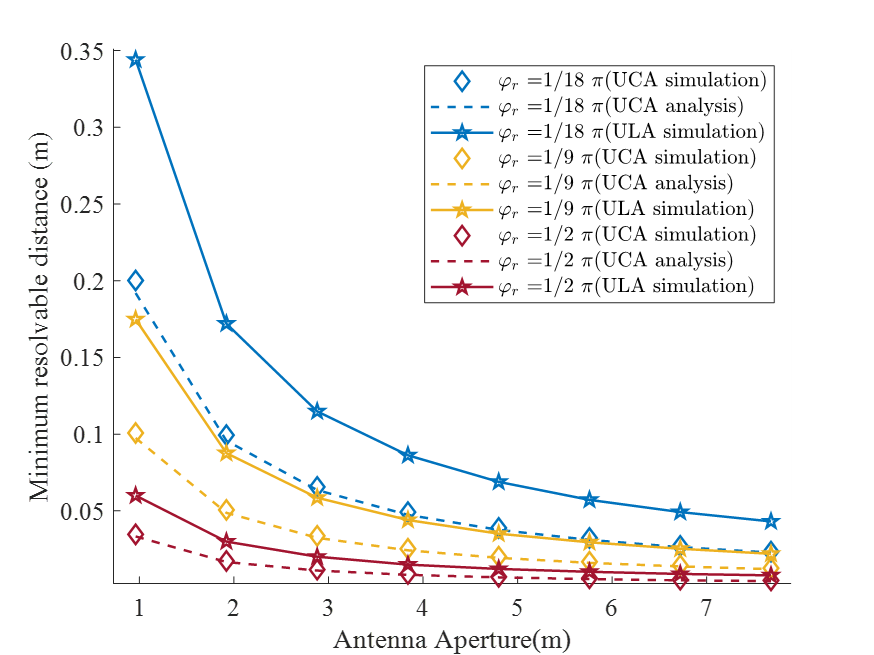}
		
	}
	\caption{\footnotesize Verification of the analytical minimum resolvable distance. }
	\label{fig:mrdvsda_UCAULA}
\end{figure}

To ensure a comparable array scale, the UCA and ULA are configured with the same number of elements and physical aperture, although their different geometries result in different element spacings.
Figs.~\ref{fig:mrdvsda_UCAULA}(a) and (b) show the minimum resolvable distance versus BS--user distance and antenna aperture, respectively, for different relative azimuths \(\varphi_r\). Here, \(\varphi_r\) denotes the angle between the displacement from the reference user to the neighboring user and the outward radial direction from the array center to the reference user. For both arrays, \(d_{\mathrm{ms}}\) increases with BS--user distance and decreases with aperture. Over the considered angles, it is smaller for \(\varphi_r\) closer to \(\pi/2\), reflecting better tangential than radial spatial resolution.
Under these matched-element and matched-aperture configurations, the UCA consistently achieves a smaller minimum resolvable distance than the ULA, indicating superior spatial resolution in near-field communications. This also implies stronger inherent interference suppression and higher spatial multiplexing capability for the UCA.
Moreover, the UCA simulation results, shown by diamond markers, closely match the closed-form predictions in~\eqref{min_d}, shown by dashed curves, with only negligible discrepancies. The solid curves with star markers report the ULA simulation results. This agreement validates the accuracy of~\eqref{min_d}.

\begin{figure}[t]
	\centering
	\subfloat[Distance-domain amplitude gain.]{
		\includegraphics[width=0.77\linewidth]{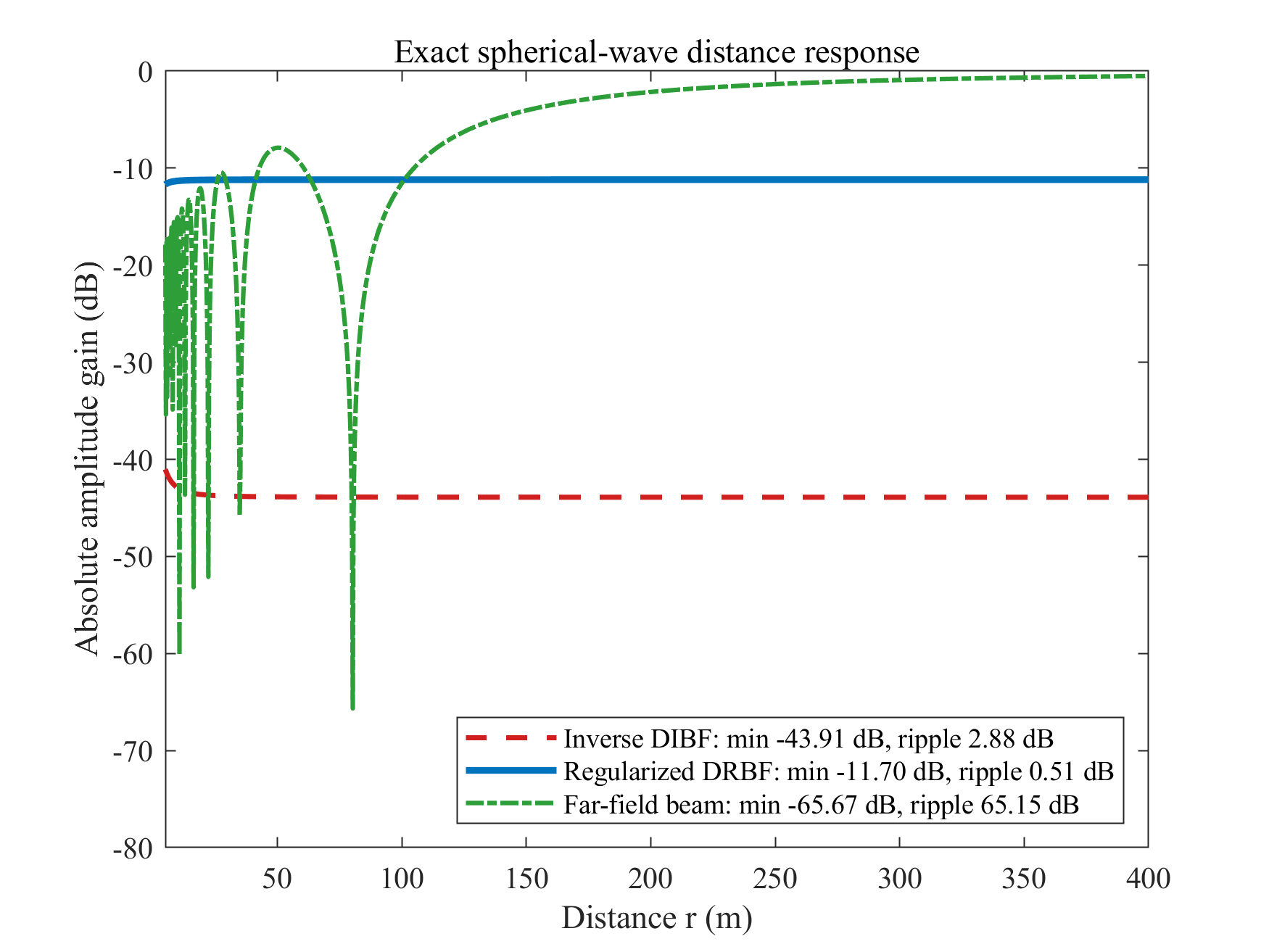}
		\label{fig:DRBF_distance}}\\[-0.5ex]
	\subfloat[Angular-domain amplitude gain.]{
		\includegraphics[width=0.77\linewidth]{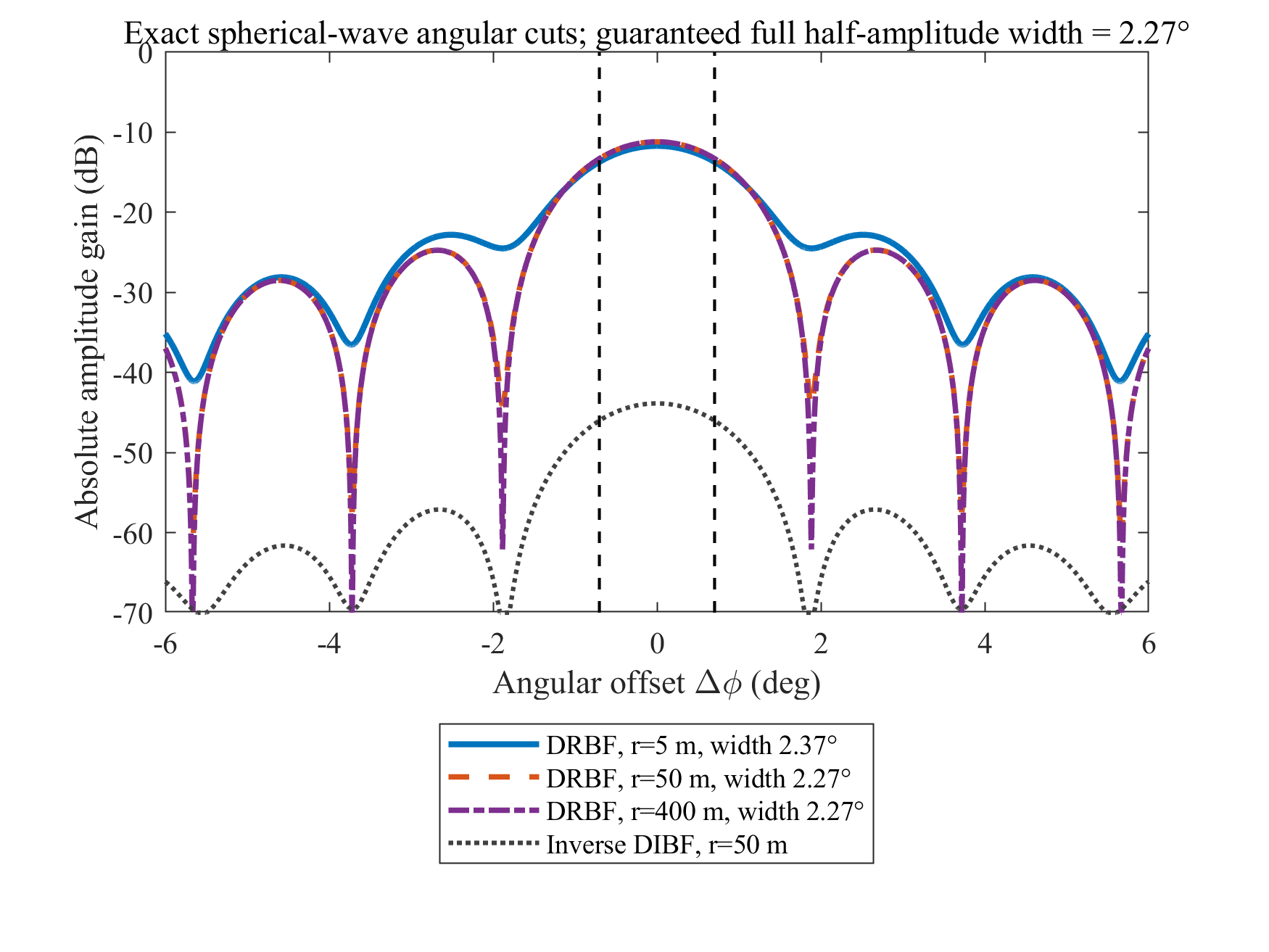}
		\label{fig:DRBF_angle}}
\caption{\footnotesize Exact-spherical absolute amplitude gains for unit-norm beams at the selected modal order $M^\star=96$: (a) distance responses of the strict inverse-modal baseline, regularized DRBF, and far-field beam; and (b) angular DRBF responses at $r=5$, $50$, and $400$~m, with the inverse-modal reference at $r=50$~m. The vertical dashed lines mark the $L=256$ sector boundaries, and no curve is pointwise normalized.}
	\label{fig:DRBF}
\end{figure}

Fig.~\ref{fig:DRBF} evaluates the selected operating point $(M^\star,\zeta_{\mathrm{rel}}^\star)=(96,39.81)$ using the exact spherical-wave manifold. In Fig.~\ref{fig:DRBF_distance}, the strict inverse-modal baseline and the proposed regularized DRBF use the same selected order $M^\star=96$, so their comparison isolates the effect of regularization. The strict inverse amplifies weak retained Bessel modes and can therefore suffer a substantial absolute-gain loss after unit-norm normalization. The proposed regularization suppresses this weak-mode amplification and improves the worst-case absolute amplitude gain while retaining a small distance variation. In contrast, the conventional far-field matched beam remains strongly affected by near-field range-dependent defocusing.

The angular cuts in Fig.~\ref{fig:DRBF_angle} confirm that the assigned sector boundaries remain within the DRBF main lobe over the considered service region. The minimum full half-amplitude width over the evaluated distances is $2.26^\circ$, approximately $10.43$ times as large as the $0.217^\circ$ width of the far-field matched beam; this broadening corresponds to a beamwidth-based power-spreading reference loss of $10.18$~dB. The evaluated worst-case on-axis loss exceeds this reference by only $1.51$~dB, while the $0.51$-dB range ripple indicates limited sensitivity to distance variations.

\begin{figure}[t]
	\centering
	\includegraphics[width=0.68\linewidth]{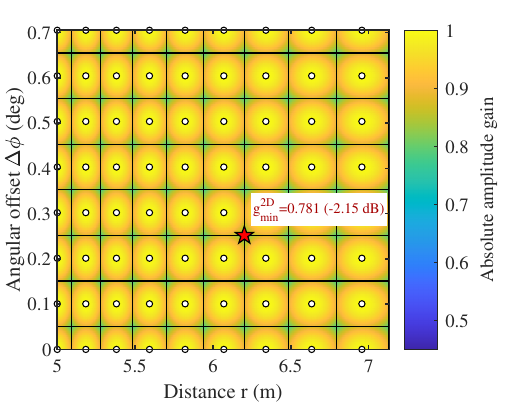}
	\caption{\footnotesize Exact-spherical 2D amplitude-gain coverage of the Layer-2 FP codebook. Black lines delineate the angular--radial cells, white circles mark codeword centers, and the red star marks the audited minimum.}
	\label{fig:FP}
\end{figure}

Fig.~\ref{fig:FP} directly evaluates the joint coverage metric in~\eqref{eq:exact_2d_coverage_metric} for the proposed Layer-2 FP codebook using the exact spherical-wave manifold and unit-norm codewords; no pointwise normalization is applied. For readability, seven adjacent angular cells and nine near-range radial samples are displayed, whereas the numerical audit spans the complete region $\mathcal D_{\mathrm{FP}}$. The maximum absolute amplitude gain remains continuous across all angular--radial cell boundaries, and the audited minimum is $g_{\min}^{\mathrm{2D}}=0.781$ ($-2.15$~dB), which exceeds the prescribed threshold $2\tau-1=0.5$ ($-6.02$~dB). This verifies joint 2D coverage without uncovered angle--distance gaps.

\begin{figure}
	\centering
	\includegraphics[width=0.8\linewidth]{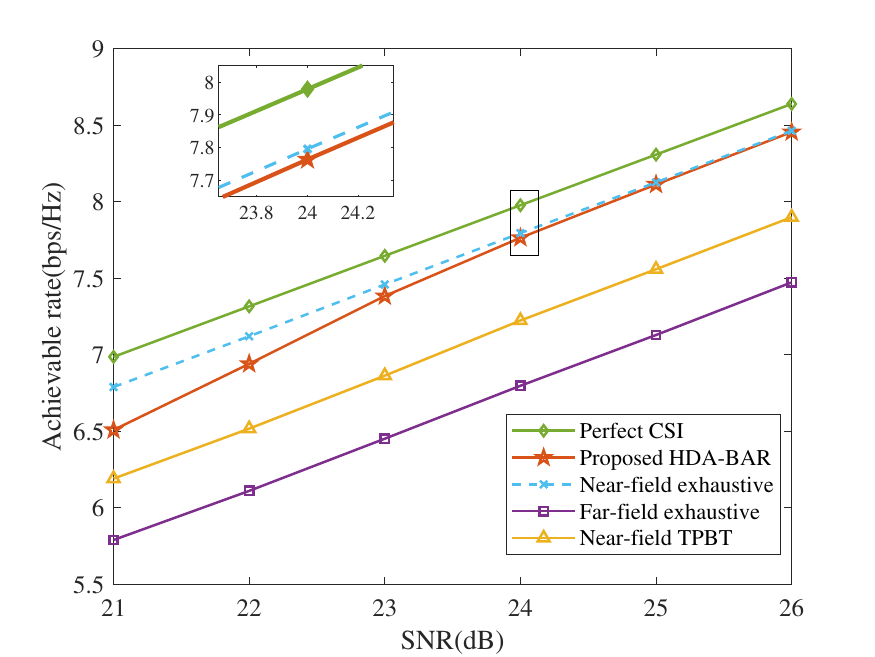}
	\caption{\footnotesize Achievable rate versus the reference SNR over $21$--$26$~dB.}
	\label{fig:achievablerates1}
\end{figure}


The achievable rate under different beam training schemes is evaluated over the reference-SNR interval $\eta_{\mathrm{ref}}\in[21,26]$~dB in Fig.~\ref{fig:achievablerates1}. 
The reference SNR is defined as
$\eta_{\mathrm{ref}}=\frac{PN|\rho|^2}{\sigma^2},$
where \(P\) is the transmit power and \(\sigma^2\) is the noise power.
Given the near-field array response vector \(\mathbf{a}(\phi,r)\) and the selected beamforming vector \(\mathbf{w}\), satisfying \(\|\mathbf{a}(\phi,r)\|_2=\|\mathbf{w}\|_2=1\), the achievable rate is given by
$R=\log_2\left(1+\eta_{\mathrm{ref}}\left|\mathbf{a}^{\mathrm H}(\phi,r)
\mathbf{w}\right|^2\right).$
The perfect-CSI scheme serves as the performance upper bound, where the beamforming vector is assumed to be perfectly aligned with the user and is given by \(\mathbf{w}=\mathbf{a}(\phi,r)\).
Under the updated $M=96$ Layer-1 configuration, the proposed two-stage search approaches the near-field exhaustive-search benchmark as the reference SNR increases and consistently outperforms near-field TPBT and far-field exhaustive search. At the lower end of the considered interval, its moderate gap to near-field exhaustive search is caused by noise-induced errors in the Stage-1 energy-correlation refinement and their propagation through candidate pruning; the gap rapidly narrows with increasing SNR. This behavior reflects the performance--overhead tradeoff of hierarchical wide-beam training rather than a loss of distance robustness.

\begin{figure}
	\centering
	\includegraphics[width=0.78\linewidth]{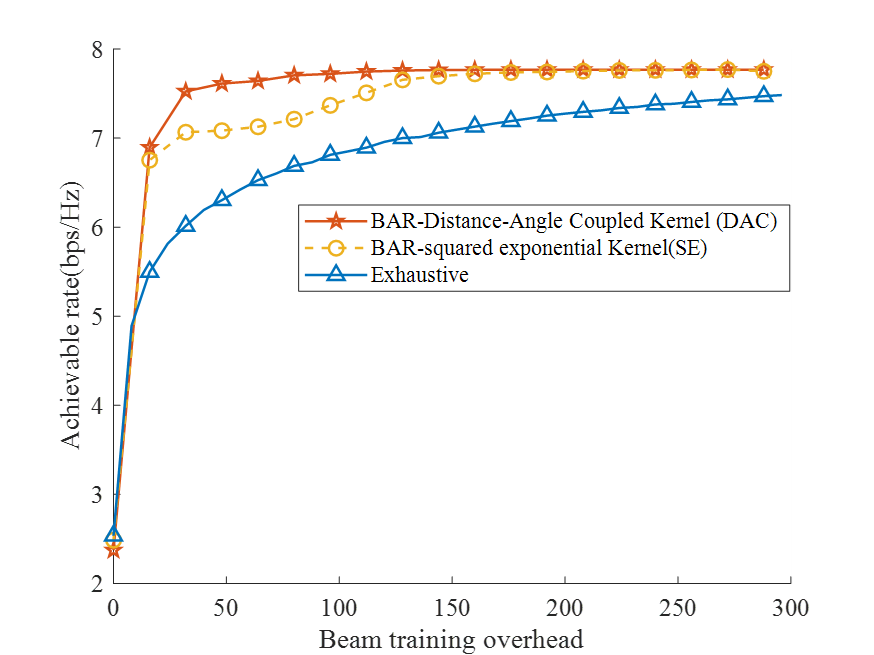}
	\caption{\footnotesize Achievable-rate comparison of the two kernels and sequential exhaustive search under limited Stage-2 training overhead.}
	\label{fig:acrvskernel}
\end{figure}

To evaluate the Stage-2 search strategy, Fig.~\ref{fig:acrvskernel} compares BAR using the proposed DAC kernel with BAR using the conventional SE kernel and sequential exhaustive search. All methods use the candidate set produced by the $M=96$ Layer-1 estimator. The exhaustive-search curve reports the best result obtained after sequentially probing the same number of codewords as the displayed training overhead, rather than the final result after scanning the entire candidate set. Within the displayed overhead range, DAC-based BAR converges faster and achieves a higher rate than SE-based BAR and sequential exhaustive search, confirming that exploiting the near-field angle--distance coupling improves the sample efficiency of Stage-2 beam training. The performance gap narrows as the overhead increases because the methods acquire increasingly informative observations.

\begin{figure}[t]
	\centering
	\subfloat[]{
		\includegraphics[width=0.78\linewidth]{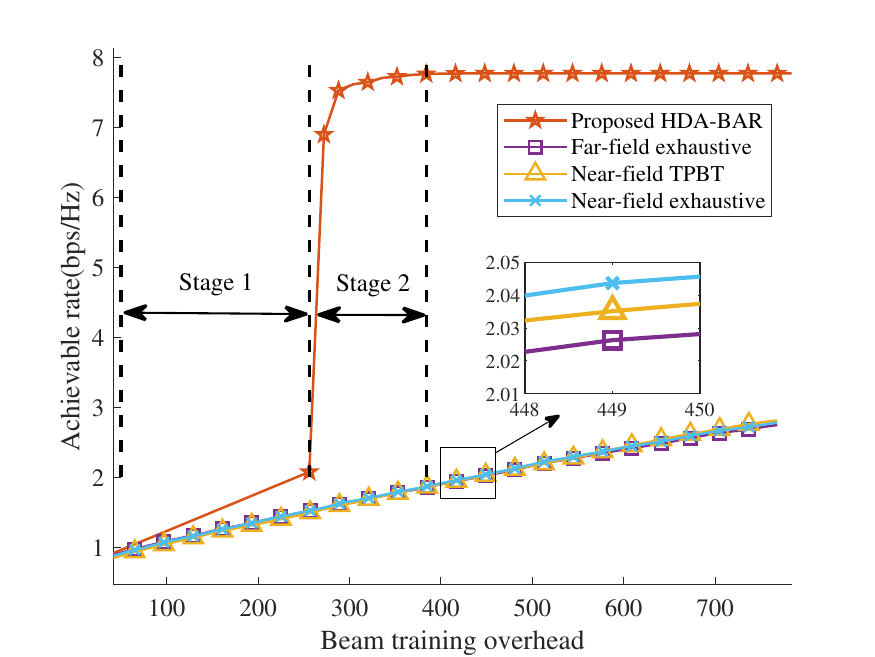}
		
	}\\
	\subfloat[]{
		\includegraphics[width=0.78\linewidth]{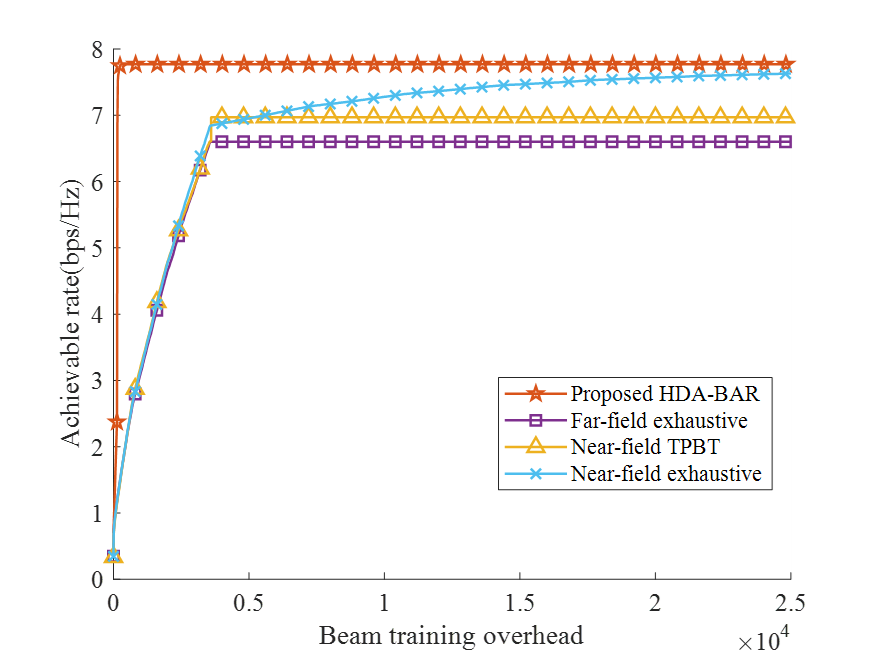}
		
	}
	\caption{\footnotesize Achievable rate performance vs. the beam training overhead.}
	\label{fig:acrvsoverhead}
\end{figure}

The achievable rate versus beam training overhead at a reference SNR of \(24\)~dB is shown in Fig.~\ref{fig:acrvsoverhead}. After $L=256$ Stage-1 probes, the proposed HDA-BAR scheme rapidly resolves the remaining angle--distance uncertainty using at most $T_{\max}=128$ additional Stage-2 probes and approaches the achievable rate of near-field exhaustive search. Consequently, HDA-BAR requires only $384$ probing slots, corresponding to a $99.66\%$ overhead reduction relative to exhaustively scanning all $N_c=114368$ FP codewords, while outperforming near-field TPBT and far-field exhaustive search. These results confirm the overall effectiveness of HDA-BAR in maintaining near-exhaustive beamforming performance with substantially reduced training overhead.

\section{Conclusion}
\label{sec:conclusion}

In this paper, we exploited the geometric properties of the UCA to investigate its spatial resolution capability in the angular and distance domains under near-field propagation.
The results showed that, with the same number of antenna elements and the same physical aperture, the UCA achieves a smaller minimum resolvable distance than the ULA, thereby providing superior near-field spatial resolution and higher spatial multiplexing potential. 
Based on this observation, a decoupled resolution-aware hierarchical codebook tailored for near-field UCA systems was proposed. 
Its Layer-1 DRBF design combines modal truncation with regularized compensation to retain distance robustness while improving the absolute amplitude gain under unit-norm transmission.
Furthermore, a low-overhead HDA-BAR beam training scheme was developed based on the proposed hierarchical codebook.
For the considered array configuration, the resulting two-stage training procedure requires $384$ probing slots and therefore reduces the training overhead by approximately 99.66\% relative to near-field exhaustive search.
As a future research direction, efficient near-field beam training for UCA systems in three-dimensional space remains a critical challenge to be further investigated.

 \section*{Appendix A}
\label{AppendixA}
Substituting the approximations of \(r_{a,n}\) in~\eqref{r1n_appromix} and \(r_{b,n}\) in~\eqref{r2n_appromix} into~\eqref{H_matrix}, the \((l,p)\)-th entry of the gain matrix \(\mathbf{G}\) can be written as
\begin{align}
	\mathbf{G}_{l,p} 
	&=
	e^{j\delta d\alpha}
	\sum_{n=1}^{N}
	e^{j\delta d\beta_1\sin\psi_n}
	\notag\\
	&\quad\times
	e^{j\delta d\beta_2\cos(2\psi_n)}
	e^{j\delta d\beta_3\sin(2\psi_n)} ,
	\label{Glp_initial}
\end{align}
where $\delta = l-p$, $\alpha =-\frac{2\pi}{\lambda}(1-\frac{R^2_t}{4r^2})\cos\varphi_r$, $\beta_1 =\frac{2\pi R_t}{r \lambda}\sin\varphi_r $, $\beta_2 =\frac{\pi R^2_t}{2r^2 \lambda }\cos\varphi_r $, $\beta_3
=\frac{\pi R_t^2}{r^2\lambda}\sin\varphi_r.$
To further simplify the above expression, we employ the Jacobi--Anger expansion~\cite{bowman2012bessel}, which can be expressed as follows
\begin{equation}
	\label{Jacobi–Anger}
	e^{j\beta\cos\vartheta}
	=
	\sum_{m=-\infty}^{\infty}
	j^{m} J_{m}(\beta)e^{jm\vartheta},
\end{equation}
where $J_m(\cdot)$ denotes the $m$-th order Bessel function of the first kind. 
Applying the Jacobi--Anger expansion, 
$\mathbf{G}_{l,k}$ can be further expressed as
Applying the Jacobi--Anger expansion, \(\mathbf{G}_{l,p}\) can be further expressed as
\begin{align}
	\label{Glp_corollary}
	\begin{split}
		\mathbf{G}_{l,p}
		=&\,
		e^{j\delta d\alpha}
		\sum_{m'=-\infty}^{\infty}
		J_{m'}\!\left(\delta d\beta_1\right)
		\sum_{m=-\infty}^{\infty}
		j^m J_m\!\left(\delta d\beta_2\right)
		\\
		&\times
		\sum_{\nu=-\infty}^{\infty}
		J_{\nu}\!\left(\delta d\beta_3\right)
		\sum_{n=1}^{N}
		e^{j(m'+2m+2\nu)\psi_n}
		\\
		&\overset{(b)}{\approx}
		N e^{j\delta d\alpha}
		\sum_{m=-\infty}^{\infty}
		\sum_{\nu=-\infty}^{\infty}
		j^m
		J_{2m+2\nu}\!\left(\delta d\beta_1\right)
		\\
		&\qquad \times
		J_m\!\left(\delta d\beta_2\right)
		J_{\nu}\!\left(\delta d\beta_3\right)
		\\
		&\overset{(c)}{\approx}
		N e^{j\delta d\alpha}
		J_0\!\left(\delta d\beta_1\right)
		J_0\!\left(\delta d\beta_2\right)
		J_0\!\left(\delta d\beta_3\right),
	\end{split}
\end{align}
where step \((b)\) follows from the discrete orthogonality property of the UCA phase modes, i.e.,
\begin{equation}
	\sum_{n=1}^{N}
	e^{j(m'+2m+2\nu)\psi_n}
	=
	\begin{cases}
		N, & m'+2m+2\nu=Nq,\ q\in\mathbb{Z},\\
		0, & m'+2m+2\nu\neq Nq,\ q\in\mathbb{Z}.
	\end{cases}
\end{equation}
This property indicates that the summation equals \(N\) only when \(m'+2m+2\nu\) is an integer multiple of the number of antennas \(N\). 
Otherwise, the summation becomes zero and thus does not contribute to the result. 
For a sufficiently large UCA and the dominant low-order phase modes, the spatial aliasing terms corresponding to \(q\neq0\) are negligible. 
Thus, the dominant contribution is obtained by taking \(q=0\), which gives $m'=-2m-2\nu$, and we have $J_{-2m-2\nu}(x)=J_{2m+2\nu}(x)$, which leads to step \((b)\).
The approximation in step \((c)\) follows from the small-argument approximation of the Bessel function. 
When \(x\rightarrow0\), the Bessel function of the first kind can be approximated as
\begin{equation}
	J_n(x)\approx\frac{1}{n!}\left(\frac{x}{2}\right)^n,\quad n\ge0 .
\end{equation}
Under the assumptions of the adopted system model, the quantities \(\delta d\beta_1\), \(\delta d\beta_2\), and \(\delta d\beta_3\) are sufficiently small. 
Consequently, the product term satisfies
\begin{equation}
	J_{2m+2\nu}(\delta d\beta_1)
	J_m(\delta d\beta_2)
	J_{\nu}(\delta d\beta_3)
	\approx 0,
	\quad (m,\nu)\neq(0,0).
\end{equation}
Retaining only the zeroth-order Bessel-function term, i.e., \(m=0\) and \(\nu=0\), leads to the approximate expression in~\eqref{Glk_appromix}. 
This completes the proof.

\section*{Appendix B}
\label{AppendixB}

We first prove Lemma~\ref{lem:drbf_range_robustness}. By rotational symmetry, set the target azimuth to zero and define
\begin{equation}
	d(\theta,r)
	=\sqrt{r^2+R_t^2-2rR_t\cos\theta}.
\end{equation}
In the dense-UCA limit, projecting the exact spherical-wave manifold onto the $m$-th phase-mode basis vector gives
\begin{equation}
	h_m(r)
	=\frac{1}{2\pi}\int_0^{2\pi}
	e^{-jk_0[d(\theta,r)-r]}e^{-jm\theta}\,\mathrm d\theta.
	\label{eq:appendix_exact_modal_integral}
\end{equation}
Its far-field limit is
\begin{equation}
	h_m(\infty)
	=\frac{1}{2\pi}\int_0^{2\pi}
	e^{jk_0R_t\cos\theta}e^{-jm\theta}\,\mathrm d\theta
	=j^mJ_m(k_0R_t).
\end{equation}
We now compare~\eqref{eq:appendix_exact_modal_integral} with this limit using the stationary-phase approximation. Around $\theta=0$,
\begin{equation}
	d(\theta,r)
	=r-R_t+\frac{rR_t}{2(r-R_t)}\theta^2
	+\mathcal O(\theta^4).
\end{equation}
Let $x=k_0R_t$ and $A_-=xr/(r-R_t)$. The phase of the integrand is locally
\begin{align}
	\Phi_m(\theta,r)
	&=-k_0[d(\theta,r)-r]-m\theta\notag\\
	&\approx x+\frac{m^2}{2A_-}
	-\frac{A_-}{2}\left(\theta+\frac{m}{A_-}\right)^2.
\end{align}
The corresponding far-field stationary phase contains $m^2/(2x)$; hence, the finite-range phase increment contributed by this stationary region is
\begin{equation}
	\Delta\Phi_{m,-}(r)
	=\frac{m^2}{2A_-}-\frac{m^2}{2x}
	=-\frac{m^2}{2k_0r}.
	\label{eq:appendix_stationary_phase_minus}
\end{equation}
Around the second stationary region, let $\theta=\pi+u$. Since
\begin{equation}
	d(\pi+u,r)
	=r+R_t-\frac{rR_t}{2(r+R_t)}u^2
	+\mathcal O(u^4),
\end{equation}
and $A_+=xr/(r+R_t)$, completing the square gives
\begin{equation}
	\Phi_m(\pi+u,r)
	\approx-x-m\pi-\frac{m^2}{2A_+}
	+\frac{A_+}{2}\left(u-\frac{m}{A_+}\right)^2.
\end{equation}
Relative to its far-field value $-m^2/(2x)$, this region produces the same increment
\begin{equation}
	\Delta\Phi_{m,+}(r)
	=-\frac{m^2}{2A_+}+\frac{m^2}{2x}
	=-\frac{m^2}{2k_0r}.
	\label{eq:appendix_stationary_phase_plus}
\end{equation}
Because both dominant stationary regions acquire the same range-dependent phase, the finite-to-far-field modal ratio, after removal of a mode-independent phase, is
\begin{equation}
	\frac{h_m(r)}{h_m(\infty)}
	=e^{-j\widetilde\delta_m(r)}
	\left[1+\varepsilon_m(r,R_t)\right],
	\qquad
	\widetilde\delta_m(r)\approx\frac{m^2}{2k_0r}.
	\label{eq:appendix_modal_ratio}
\end{equation}
Here, $\varepsilon_m$ collects stationary-phase amplitude variation, higher-order path terms, and finite-array modal aliasing. Equation~\eqref{eq:appendix_modal_ratio} applies to retained modes away from zeros of $J_m(k_0R_t)$; weak modes are subsequently controlled by regularization.

Ignoring $\varepsilon_m$ in the leading-order analysis, strict inverse-modal equalization gives the normalized on-axis response
\begin{equation}
	A_M(r)
	=\frac{1}{2M+1}
	\sum_{m=-M}^{M}e^{-j\widetilde\delta_m(r)}.
\end{equation}
Its squared amplitude satisfies
\begin{align}
	|A_M(r)|^2
	&=\frac{1}{(2M+1)^2}
	\sum_{m=-M}^{M}\sum_{\ell=-M}^{M}
	\cos\!\left(\widetilde\delta_m-\widetilde\delta_\ell\right)\notag\\
	&\approx
	1-\operatorname{Var}(\widetilde\delta_m)
	=1-\frac{\operatorname{Var}(m^2)}{4(k_0r)^2},
\end{align}
where $\cos z=1-z^2/2+\mathcal O(z^4)$ has been used. Therefore,
\begin{equation}
	|A_M(r)|
	\approx1-\frac{\operatorname{Var}(m^2)}{8(k_0r)^2}.
\end{equation}
For uniformly weighted modes $m=-M,\ldots,M$,
\begin{align}
	\mathbb E[m^2]
	&=\frac{M(M+1)}{3},
\end{align}
\begin{equation}
	\operatorname{Var}(m^2)
	=\frac{M(M+1)[4M(M+1)-3]}{45}.
\end{equation}
Substitution gives~\eqref{eq:inverse_modal_range_loss}. Finally, comparing $\widetilde\delta_M(r)$ with the maximum curvature phase $k_0R_t^2/(2r)$ in~\eqref{eq:drbf_range_residual} gives the factor $(M/(k_0R_t))^2$, completing the proof of Lemma~\ref{lem:drbf_range_robustness}.

%

\bibliographystyle{IEEEtran}  
\bibliography{bib1_modified_abbrev.bib}           

@ARTICLE{Cui2022ULA,
	author={Cui, Mingyao and Dai, Linglong},
	journal={IEEE Trans. Commun.}, 
	title={Channel Estimation for Extremely Large-Scale {MIMO}: Far-Field or Near-Field?}, 
	year={2022},
	volume={70},
	number={4},
	pages={2663--2677},
	doi={10.1109/TCOMM.2022.3146400}}

@ARTICLE{Song2024,
	author={Song, Ruihao and Shen, Jun and Yuan, Hang and Gao, Xiaozheng and Ding, Xuhui and Liu, Yuanwei and da Costa, Daniel Benevides},
	journal={IEEE Internet Things J.}, 
	title={Line-of-Sight {MIMO} Systems: {DoF} Analysis and Hybrid Beamforming Design}, 
	year={2024},
	volume={11},
	number={19},
	pages={31791--31804},
	doi={10.1109/JIOT.2024.3423802}}

@ARTICLE{Liu2023,
	author={Xie, Ziyi and Liu, Yuanwei and Xu, Jiaqi and Wu, Xuanli and Nallanathan, Arumugam},
	journal={IEEE Wireless Commun. Lett.}, 
	title={Performance Analysis for Near-Field {MIMO}: Discrete and Continuous Aperture Antennas}, 
	year={2023},
	volume={12},
	number={12},
	pages={2258--2262},
	doi={10.1109/LWC.2023.3317492}}

@ARTICLE{Yuan2023,
	author={Yuan, Zhiqiang and Zhang, Fengchun and Zhang, Yuxiang and Zhang, Jianhua and Pedersen, Gert Frølund and Fan, Wei},
	journal={IEEE Trans. Antennas Propag.}, 
	title={On Phase Mode Selection in the Frequency-Invariant Beamformer for Near-Field {mmWave} Channel Characterization}, 
	year={2023},
	volume={71},
	number={11},
	pages={8975--8986},
	doi={10.1109/TAP.2023.3316791}}

@INPROCEEDINGS{Ly2024,
	author={Nguyen, Ly V. and Nguyen, Duy H. N. and Atzeni, Italo and Tölli, Antti and Swindlehurst, A. Lee},
	booktitle={Proc. IEEE Sensor Array Multichannel Signal Process. Workshop (SAM)}, 
	title={Channel Estimation in Low-Resolution Near-Field Massive {MIMO} Systems}, 
	year={2024},
	volume={},
	number={},
	pages={1--5},
	doi={10.1109/SAM60225.2024.10636448}}

@ARTICLE{Ray2021,
	author={Ray, Partha Pratim and Kumar, Neeraj and Guizani, Mohsen},
	journal={IEEE Wireless Commun.}, 
	title={A Vision on {6G}-Enabled {NIB}: Requirements, Technologies, Deployments, and Prospects}, 
	year={2021},
	volume={28},
	number={4},
	pages={120--127},
	doi={10.1109/MWC.001.2000384}}

@ARTICLE{Dong2025,
	author  = {Dong, Lijun and Hu, Nan and Deng, Wei and Xu, Xiaodong and Ding, Haiyu and Huang, Yuhong},
	journal = {IEEE Wireless Commun.},
	title   = {Quality of Service Evolution and Enhancement for Holographic Video Communications Toward {6G}},
	year    = {2025},
	volume  = {32},
	number  = {2},
	pages   = {140--146},
	doi     = {10.1109/MWC.001.2400202}
}

@ARTICLE{Zhang2024,
	author={Zhang, Ronghui and He, Yuan and Yuan, Lufeng and Li, Ziyue and He, Chendi and Tan, Fangqing},
	journal={IEEE Wireless Commun.}, 
	title={Empowering {6G} Ambient Intelligence with Terahertz Integrated Sensing and Communications}, 
	year={2024},
	volume={31},
	number={5},
	pages={256--263},
	doi={10.1109/MWC.017.2300532}}

@ARTICLE{Wang2023,
	author={Wang, Cheng-Xiang and You, Xiaohu and Gao, Xiqi and Zhu, Xiuming and Li, Zixin and Zhang, Chuan and Wang, Haiming and Huang, Yongming and Chen, Yunfei and Haas, Harald and Thompson, John S. and Larsson, Erik G. and Renzo, Marco Di and Tong, Wen and Zhu, Peiying and Shen, Xuemin and Poor, H. Vincent and Hanzo, Lajos},
	journal={IEEE Commun. Surveys Tuts.}, 
	title={On the Road to {6G}: Visions, Requirements, Key Technologies, and Testbeds}, 
	year={2023},
	volume={25},
	number={2},
	pages={905--974},
	doi={10.1109/COMST.2023.3249835}}

@ARTICLE{Mingyao2023,
	author={Cui, Mingyao and Wu, Zidong and Lu, Yu and Wei, Xiuhong and Dai, Linglong},
	journal={IEEE Commun. Mag.}, 
	title={Near-Field {MIMO} Communications for {6G}: Fundamentals, Challenges, Potentials, and Future Directions}, 
	year={2023},
	volume={61},
	number={1},
	pages={40--46},
	doi={10.1109/MCOM.004.2200136}}

@ARTICLE{Wei2022,
	author={Wei, Xiuhong and Dai, Linglong},
	journal={IEEE Commun. Lett.}, 
	title={Channel Estimation for Extremely Large-Scale Massive {MIMO}: Far-Field, Near-Field, or Hybrid-Field?}, 
	year={2022},
	volume={26},
	number={1},
	pages={177--181},
	doi={10.1109/LCOMM.2021.3124927}}

@ARTICLE{Zhang2022,
	author={Zhang, Haiyang and Shlezinger, Nir and Guidi, Francesco and Dardari, Davide and Imani, Mohammadreza F. and Eldar, Yonina C.},
	journal={IEEE Trans. Wireless Commun.}, 
	title={Beam Focusing for Near-Field Multiuser {MIMO} Communications}, 
	year={2022},
	volume={21},
	number={9},
	pages={7476--7490},
	doi={10.1109/TWC.2022.3158894}}

@ARTICLE{Lu2023,
	author={Lu, Yu and Dai, Linglong},
	journal={IEEE Trans. Commun.}, 
	title={Near-Field Channel Estimation in Mixed {LoS}/{NLoS} Environments for Extremely Large-Scale {MIMO} Systems}, 
	year={2023},
	volume={71},
	number={6},
	pages={3694--3707},
	doi={10.1109/TCOMM.2023.3260242}}

@ARTICLE{Sun2025,
	author={Sun, Shu and Li, Renwang and Han, Chong and Liu, Xingchen and Xue, Liuxun and Tao, Meixia},
	journal={IEEE Commun. Mag.}, 
	title={How to Differentiate Between Near Field and Far Field: Revisiting the Rayleigh Distance}, 
	year={2025},
	volume={63},
	number={1},
	pages={22--28},
	doi={10.1109/MCOM.001.2400007}}

@ARTICLE{Lee2016,
	author={Lee, Junho and Gil, Gye-Tae and Lee, Yong H.},
	journal={IEEE Trans. Commun.}, 
	title={Channel Estimation via Orthogonal Matching Pursuit for Hybrid {MIMO} Systems in Millimeter Wave Communications}, 
	year={2016},
	volume={64},
	number={6},
	pages={2370--2386},
	doi={10.1109/TCOMM.2016.2557791}}

@ARTICLE{Hou2025,
	author={Hou, Xiaolin and Björnson, Emil and Lee, Namyoon and Heath, Robert W.},
	journal={IEEE Commun. Mag.}, 
	title={Guest Editorial: Near-Field {MIMO} Technologies Toward {6G}}, 
	year={2025},
	volume={63},
	number={1},
	pages={20--21},
	doi={10.1109/MCOM.2025.10819471}}

@ARTICLE{CodebookWei2022,
	author={Wei, Xiuhong and Dai, Linglong and Zhao, Yajun and Yu, Guanghui and Duan, Xiangyang},
	journal={China Commun.}, 
	title={Codebook design and beam training for extremely large-scale {RIS}: Far-field or near-field?}, 
	year={2022},
	volume={19},
	number={6},
	pages={193--204},
	doi={10.23919/JCC.2022.06.015}}

@ARTICLE{TrainingZhang2022,
	author={Zhang, Yunpu and Wu, Xun and You, Changsheng},
	journal={IEEE Wireless Commun. Lett.}, 
	title={Fast Near-Field Beam Training for Extremely Large-Scale Array}, 
	year={2022},
	volume={11},
	number={12},
	pages={2625--2629},
	doi={10.1109/LWC.2022.3212344}}

@ARTICLE{Lu2024,
	author={Lu, Yu and Zhang, Zijian and Dai, Linglong},
	journal={IEEE Trans. Commun.}, 
	title={Hierarchical Beam Training for Extremely Large-Scale {MIMO}: From Far-Field to Near-Field}, 
	year={2024},
	volume={72},
	number={4},
	pages={2247--2259},
	doi={10.1109/TCOMM.2023.3344600}}

@ARTICLE{Yuanwei2024,
	author={Wu, Chenyu and You, Changsheng and Liu, Yuanwei and Chen, Li and Shi, Shuo},
	journal={IEEE Trans. Veh. Technol.}, 
	title={Two-Stage Hierarchical Beam Training for Near-Field Communications}, 
	year={2024},
	volume={73},
	number={2},
	pages={2032--2044},
	doi={10.1109/TVT.2023.3311868}
	}

@ARTICLE{Shi2024,
	author={Shi, Xu and Wang, Jintao and Sun, Zhi and Song, Jian},
	journal={IEEE Trans. Wireless Commun.}, 
	title={Spatial-Chirp Codebook-Based Hierarchical Beam Training for Extremely Large-Scale Massive {MIMO}}, 
	year={2024},
	volume={23},
	number={4},
	pages={2824--2838},
	doi={10.1109/TWC.2023.3303229}
}

@ARTICLE{Qi2022,
	author={Qi, Chenhao and Chen, Kangjian and Dobre, Octavia A. and Li, Geoffrey Ye},
	journal={IEEE Trans. Wireless Commun.}, 
	title={Hierarchical Codebook-Based Multiuser Beam Training for Millimeter Wave Massive {MIMO}}, 
	year={2020},
	volume={19},
	number={12},
	pages={8142--8152},
	doi={10.1109/TWC.2020.3019523}}

@ARTICLE{Zhuo2025,
	author={Xu, Zhuo and Zhang, Zijian and Dai, Linglong},
	journal={IEEE Trans. Wireless Commun.}, 
	title={Near-Optimal Near-Field Beam Training: From Searching to Inference}, 
	year={2025},
	volume={24},
	number={11},
	pages={9173--9185},
	doi={10.1109/TWC.2025.3571488}}

@ARTICLE{Xie2024UCA,
	author={Xie, Yuxin and Ning, Boyu and Li, Lingxiang and Chen, Zhi},
	journal={IEEE Wireless Commun. Lett.}, 
	title={Near-Field Beam Training in {THz} Communications: The Merits of Uniform Circular Array}, 
	year={2023},
	volume={12},
	number={4},
	pages={575--579},
	doi={10.1109/LWC.2023.3234001}}

@ARTICLE{Guo2024UCA,
	author={Guo, Yunhui and Zhang, Yang and Wang, Zhaolin and Liu, Yuanwei},
	journal={IEEE Trans. Wireless Commun.}, 
	title={Wideband Beamforming for Near-Field Communications With Circular Arrays}, 
	year={2024},
	volume={23},
	number={12},
	pages={19065--19082},
	doi={10.1109/TWC.2024.3478368}}

@ARTICLE{Wu2024UCA,
	author={Wu, Zidong and Cui, Mingyao and Dai, Linglong},
	journal={IEEE Trans. Wireless Commun.}, 
	title={Enabling More Users to Benefit From Near-Field Communications: From Linear to Circular Array}, 
	year={2024},
	volume={23},
	number={4},
	pages={3735--3748},
	doi={10.1109/TWC.2023.3310912}}

@ARTICLE{Ding2023UCA,
	author={Ding, Zihang and Zhang, Jianhua and Yuan, Zhiqiang and Tang, Pan and You, Changsheng and Tian, Lei and Chen, Ke and Liu, Guangyi},
	journal={IEEE Trans. Veh. Technol.}, 
	title={Low-Dimensional Near-Field Codebook Design for Extremely Large-Scale {MIMO} Systems With Uniform Circular Arrays}, 
	year={2025},
	volume={74},
	number={10},
	pages={16511--16515},
	doi={10.1109/TVT.2025.3568431}}

@ARTICLE{Zhang2019FIBF,
	author={Zhang, Fengchun and Fan, Wei},
	journal={IEEE Trans. Veh. Technol.}, 
	title={Near-Field Ultra-Wideband {mmWave} Channel Characterization Using Successive Cancellation Beamspace {UCA} Algorithm}, 
	year={2019},
	volume={68},
	number={8},
	pages={7248--7259},
	doi={10.1109/TVT.2019.2926783}}

@ARTICLE{CuiDai2024,
	author={Cui, Mingyao and Dai, Linglong},
	journal={IEEE Trans. Wireless Commun.}, 
	title={Near-Field Wideband Beamforming for Extremely Large Antenna Arrays}, 
	year={2024},
	volume={23},
	number={10},
	pages={13110--13124},
	doi={10.1109/TWC.2024.3398770}}

@INPROCEEDINGS{Roy2002,
	author={Roy, Olivier and Vetterli, Martin},
	booktitle={Proc. 15th Eur. Signal Process. Conf. (EUSIPCO)}, 
	title={The effective rank: A measure of effective dimensionality}, 
	year={2007},
	volume={},
	number={},
	pages={606--610},
	doi={}}

@book{bowman2012bessel,
	title     = {Introduction to Bessel Functions},
	author    = {Bowman, Frank},
	publisher = {Courier Corporation},
	address   = {North Chelmsford, U.K.},
	year      = {2012}
}

@ARTICLE{Rappaport2019,
	author={Rappaport, Theodore S. and Xing, Yunchou and Kanhere, Ojas and Ju, Shihao and Madanayake, Arjuna and Mandal, Soumyajit and Alkhateeb, Ahmed and Trichopoulos, Georgios C.},
	journal={IEEE Access}, 
	title={Wireless Communications and Applications Above {100 GHz}: Opportunities and Challenges for {6G} and Beyond}, 
	year={2019},
	volume={7},
	number={},
	pages={78729--78757},
	doi={10.1109/ACCESS.2019.2921522}}

@book{Bishop2006,
	author    = {C. M. Bishop},
	title     = {Pattern Recognition and Machine Learning},
	publisher = {Springer},
	address   = {New York, NY, USA},
	year      = {2006}
}

@INPROCEEDINGS{Zijian2024,
	author={Zhang, Zijian and Zhu, Jieao and Dai, Linglong and Heath, Robert W.},
	booktitle={Proc. IEEE Wireless Commun. Netw. Conf. (WCNC)}, 
	title={Successive Bayesian Reconstructor for FAS Channel Estimation}, 
	year={2024},
	volume={},
	number={},
	pages={1--5},
	doi={10.1109/WCNC57260.2024.10571302}}

@book{Rasmussen2006GP,
	title={Gaussian Processes for Machine Learning},
	author={Rasmussen, Carl Edward and Williams, Christopher},
	year={2006},
	publisher={MIT Press}
}

@article{Cui2023NearFieldWideband,
	title   = {Near-Field Wideband Channel Estimation for Extremely Large-Scale {MIMO}},
	author  = {Cui, Ming and Dai, Linglong},
	journal = {Science China Information Sciences},
	volume  = {66},
	number  = {17},
	pages   = {172303},
	year    = {2023},
	doi     = {10.1007/s11432-022-3654-y}
}

@ARTICLE{Bacci2023,
	author={Bacci, Giacomo and Sanguinetti, Luca and Björnson, Emil},
	journal={IEEE Wireless Commun. Lett.}, 
	title={Spherical Wavefronts Improve {MU-MIMO} Spectral Efficiency When Using Electrically Large Arrays}, 
	year={2023},
	volume={12},
	number={7},
	pages={1219--1223},
	doi={10.1109/LWC.2023.3268087}}

@ARTICLE{Ruiz2024,
	author={Ruiz-Sicilia, Juan Carlos and Renzo, Marco Di and Mursia, Placido and Kaushik, Aryan and Sciancalepore, Vincenzo},
	journal={IEEE Trans. Green Commun. Netw.}, 
	title={Spatial Multiplexing in Near-Field Line-of-Sight {MIMO} Communications: Paraxial and Non-Paraxial Deployments}, 
	year={2025},
	volume={9},
	number={1},
	pages={338--353},
	doi={10.1109/TGCN.2024.3418842}}

@ARTICLE{Huang2025EnergyEfficient,
		author  = {Huang, Yuting and Gao, Xiaozheng and Shi, Minwei and Ye, Neng and Yang, Kai},
		journal = {IEEE Trans. Veh. Technol.},
		title   = {Energy-efficient power control in {D2D} networks: A distributed {ADMM} approach with dynamic penalty coefficient},
		year    = {2025},
		volume  = {74},
		number  = {5},
		pages   = {8238--8250},
		doi     = {10.1109/TVT.2025.3530613}
	}

@ARTICLE{Song2026NearField,
	author  = {Song, Ruihao and Yuan, Hang and Song, Chenran and Zhang, Zeyu and Gao, Xiaozheng and Niyato, Dusit and Yang, Kai},
	journal = {IEEE Trans. Commun.},
	title   = {Near-field {LoS} {MIMO} with dual continuous apertures ({CAPs}): Channel decomposition and {EDoF}-optimal beamforming},
	year    = {2026},
	volume  = {74},
	pages   = {2251--2267},
	doi     = {10.1109/TCOMM.2025.3644485}
}

@ARTICLE{Song2026NearFieldTHzCovert,
	author  = {Song, Chenran and Yuan, Hang and Zhang, Zeyu and Gao, Xiaozheng and Shi, Minwei and Yang, Nan and Yang, Kai},
	journal = {IEEE Trans. Wireless Commun.},
	title   = {Near-field terahertz covert communications with noise uncertainty},
	year    = {2026},
	volume  = {25},
	pages   = {7520--7534},
	doi     = {10.1109/TWC.2025.3631844}
}

@ARTICLE{Shen2026FullyDynamic,
	author  = {Shen, Jun and Yuan, Hang and Gao, Xiaozheng and Ye, Neng and {Benevides da Costa}, Daniel and Yang, Kai},
	journal = {Digit. Commun. Netw.},
	title   = {Fully-dynamic and hardware-efficient analog architecture for terahertz wideband hybrid beamforming with multi-channel low-resolution phase shifters},
	year    = {2026}
}

@ARTICLE{Zhao2026AIEmpowered,
	author  = {Zhao, Yue and Li, Zan and Zhao, Changyuan and Zhang, Yiteng and Niyato, Dusit},
	journal = {IEEE Wireless Commun.},
	title   = {Toward {AI}-empowered wireless position sensing: An integrated framework of perception, cognition, decision, and learning},
	year    = {2026},
	pages   = {1--9},
	doi     = {10.1109/MWC.2026.3690107}
}

\end{document}